\documentclass[12pt,twoside,a4paper]{article}
     
\newcommand{\ggg}{\left(2-2\alpha\right)}
\newcommand{\besterNenner}{{-\ln (1+ \delta)+    \delta}}

\usepackage[OT4]{fontenc}
\usepackage[cp1250]{inputenc}
\usepackage[dvips]{color}
\usepackage{amsfonts}
\usepackage{amsmath}
\usepackage{amsthm}
\usepackage{graphicx}
\usepackage{hyperref}
\usepackage{float}
\usepackage{url}
\usepackage{array}
\usepackage{algorithm,algpseudocode}

\newcommand{\Bem}[1]{}

\newtheorem{theorem}{Theorem}
\newtheorem{lemma}{Lemma}
\def\numx#1e#2{{#1}\mathrm{e}{#2}} 

\begin{document}
\newcommand{\MovjTytulv}{Machine Learning Friendly Set Version of Johnson-Lindenstrauss Lemma}
\newcommand{\MaInstytucja}{Institute of Computer Science of the Polish Academy of Sciences\\ul. Jana Kazimierza 5, 01-248 Warszawa
Poland}
\title{
\MovjTytulv  }
\author{
Mieczys{\l}aw A. K{\l}opotek (\url{klopotek@ipipan.waw.pl})  
\\ \MaInstytucja
}

\newenvironment{keywords}{%
\noindent{\bf Keywords:}
}
  
\maketitle

\begin{abstract}
In this paper we make a novel use of the Johnson-Lindenstrauss Lemma.
The Lemma has an existential form saying that there exists 
a JL transformation $f$ of the data points into  lower dimensional space such that all of them fall into predefined error range $\delta$. 

We formulate in this paper a theorem stating that we can choose the target dimensionality in a random projection type JL linear transformation in such a way that with probability $1-\epsilon$  all of them fall into predefined error range $\delta$ for any user-predefined failure probability $\epsilon$. 
 
This result is important for applications such a data clustering where we want to have a priori dimensionality reducing transformation instead of trying out a (large) number of them, as with traditional Johnson-Lindenstrauss Lemma.
In particular, we take a closer look at the $k$-means algorithm and prove that 
a good solution in the projected space is also a good solution in the original space. 
Furthermore, under proper assumptions local optima in the original space 
are also ones in the projected space.
We define also conditions for which clusterability property of the original space is transmitted to the projected space, so that special case algorithms for the original space are also applicable in the projected space.        
 \end{abstract}
\noindent
\pagebreak 

\section{Introduction}

Dimensionality reduction plays an important role in many areas of data processing, and especially in machine learning (cluster analysis, classifier learning, model validation, data visualisation etc.). 

Usually it is associated with manifold learning, that is a belief that the data lie in fact in a low dimensional subspace that needs to be identified and the data projected onto it so that the number of degrees of freedom is reduced and as a consequence also sample sizes can be smaller without loss of reliability. Techniques like reduced $k$-means \cite{Terada:2014}, 
 PCA (Principal Component Analysis),
 Kernel PCA,
LLE (Locally Linear Embedding), 
  LEM (Laplacian Eigenmaps),
 MDS (Metric Multidimensional Scaling), 
 Isomap,
SDE (Semidefinite Embedding), 
just to mention a few. 
\Bem{
http://www.stat.washington.edu/courses/stat539/spring14/Resources/tutorial_nonlin-dim-red.pdf
}

But there exists still another possibility of approaching the dimensionality reduction problems, in particular when such intrinsic subspace where data is located cannot be identified. 
The problem of choice of the subspace has been surpassed by several authors by so-called random projection, applicable in particularly highly dimensional spaces (tens of thousands of dimensions) and correspondingly large data sets (of at least hundreds of points).   

The starting point here is the Johnson-Lindenstrauss Lemma~\cite{Johnson:1982}. 
Roughly speaking it states that there exists a linear\footnote{JL Lemma speaks about a general transformation, but many researchers look just for linear ones.} mapping from a higher dimensional space   into a sufficiently high dimensional subspace that will preserve approximately the distances between points, as needed e.g. by $k$-means algorithm \cite{CLU:AV07}. 

To be more formal 
consider a set $\mathfrak{Q}$ of $m$ objects $\mathfrak{Q}=\{1,\dots,m\}$. 
An object $i\in \mathfrak{Q}$   may have a  representation $\mathbf{x_i}\in \mathbb{R}^n$. Then the set of these representations will be denoted by $Q$. 
An object $i\in \mathfrak{Q}$   may have a  representation  $\mathbf{x'_i}\in \mathbb{R}^{n'}$, in a different space.  Then the set of these representations will be denoted by $Q'$. 

 With this notation let us state:

\begin{theorem}{(Johnson-Lindenstrauss)} \label{thm:JL_lemma}
Let $\delta \in (0,\frac12)$.
Let $\mathfrak{Q}$ be a set of  $m$ objects
and $Q$ 
 - a set of points  representing them in  $\mathbb{R}^n$, 
and let $n'\ge\frac{C \ln m}{\delta^2}$, where $C$ is a sufficiently large constant (e.g.20). 
  There
exists a Lipschitz mapping
$f: \mathbb{R}^n\rightarrow \mathbb{R}^{n'}$ such that for all $\mathbf{u},\mathbf{v} \in Q$
\begin{equation}\label{eq:JL}
(1-\delta)\| \mathbf{u}-\mathbf{v}\|^2 \le \|f(\mathbf{u})-f(\mathbf{v})\|^2  \le  (1+\delta)\| \mathbf{u}-\mathbf{v}\|^2 
\end{equation}
\end{theorem}



%





A number of proofs and applications of this theorem have been proposed which in fact do not prove the theorem as such but rather create a probabilistic version of it, like e.g. \cite{Dasgupta:2003,Achlioptas:2003,Ailon:2006,Indyk:2007,Larsen:2016,Chiong:2016}.
For an overview of Johnson-Lindenstrauss Lemma variants see 
e.g. \cite{Matousek:2008}. 
%

Essentially the idea behind these probabilistic proofs is as follows:
It is proven that 
the probability of reconstructing the length of a random vector
from a projection onto a subspace 
within a reasonable error boundaries is high. 

One then inverts the thinking and states that 
the probability of reconstructing the length of a given vector
from a projection onto a (uniformly selected) random subspace 
within a reasonable error boundaries is high. 

But uniform sampling of high dimensional subspaces is a hard task. So instead $n'$ vectors with random coordinates are sampled from the original $n$-dimensional space and one uses them as a coordinate system in the $n'$-dimensional subspace which is a much simpler process. One hopes that the sampled vectors will be orthogonal (and hence the coordinate system will be orthogonal) which in case of vectors with thousands of coordinates is reasonable. 
That means we create a matrix $M$ of $n'$ rows and $n$ columns
as follows: for each row $i$ we sample $n$ numbers from $\mathcal{N}(0,1)$ forming a row vector $\mathbf{a}^T_i$. We normalize it obtaining the row vector 
$\mathbf{b}^T_i=\mathbf{a}^T_i \cdot (\mathbf{a}^T_i\mathbf{a}_i)^{-1/2}$. This becomes the $i$th row of the matrix $M$. 
Then for any data point $\mathbf{x}$ in the original space  
its random projection is obtained as $ \mathbf{x'}=M\mathbf{x}$. 

Then the mapping we seek is the projection multiplied by a suitable factor. 

It is claimed afterwards that this mapping is distance-preserving not only for a single vector, but also for large sets of points with some, usually very small probability, as Dasgupta and Gupta  \cite{Dasgupta:2003} maintain. Via  applying the above process many times one can finally get the mapping $f$ that is needed.
That is each time we sample a subspace from the space of subspaces and check if condition expressed by equation 
(\ref{eq:JL})  holds for all the points, and if not,  we sample again, while we have the reasonable hope that we will get the subspace of interest after a finite number of steps with probability that we assume.

In this paper we explore the following flaw of the mentioned approach: 
If we want to apply for example a $k$-means clustering algorithm, we are in fact not interested in resampling the subspaces in order to find a convenient one so that the distances are sufficiently preserved.
Computation over and over again of $m^2/2$ distances between the points in the projected space may turn out to be much more expensive than computing $O(mk)$ distances during $k$-means clustering (if $m \gg k$) in the original space. In fact we are primarily interested in clustering data.  
But we do not have any criterion  for the $k$-means algorithm that would say that this particular subspace  is the right one via e.g. minimization of $k$-means criterion (and in fact for any other clustering algorithm). 

Therefore, we rather seek a scheme that will allow us to say that by a certain random sampling we have already found the subspace that we sought with a sufficiently high probability. 
As far as we know, this is the first time such a problem has been posed. 

To formulate claims concerning $k$-means, we need to introduce additional notation. 
Let us denote with $\mathfrak{C}$ a partition of $\mathfrak{Q}$ into $k$ clusters $\{C_1,\dots,C_k\}$.  
For any $i\in \mathfrak{Q}$ let $\mathfrak{C}(i)$ denote the cluster $C_j$ to which $i$ belongs. 
For any set of objects $C_j$ let 
$\boldsymbol\mu(C_j)=\frac 1 {|C_j|} \sum_{i\in C_j} \mathbf{x_i}$ and
$\boldsymbol\mu'(C_j)=\frac 1 {|C_j|} \sum_{i\in C_j} \mathbf{x'_i}$.

Under this notation the $k$-means  cost function may be written as 
\begin{align}
 \mathfrak{J}(Q, \mathfrak{C})=\sum_{ i\in \mathfrak{Q} }\|\mathbf{x }_i-\boldsymbol\mu(\mathfrak{C}(i)) \|^2
\label{eq:defJ}
\\  %
 \mathfrak{J}(Q',\mathfrak{C})=\sum_{ i\in \mathfrak{Q} }\|\mathbf{x'}_i-\boldsymbol\mu'(\mathfrak{C}(i) \|^2
\label{eq:defJprim}
\end{align}
for the sets $Q,Q'$.

Our contribution is as follows:
\begin{itemize}
\item We formulate and prove a set version of JL Lemma - see Theorem \ref{thm:Klvopotek_Lemma}.
\item Based on it we demonstrate that a good solution to $k$-means problem in the projected space is also a good one in the original space 
- see Theorem \ref{thm:Klvopotek_k_means_projected_cost}.
\item We show that local  $k$-means  minima in the original and the projected spaces match under proper conditions
- see Theorems \ref{thm:Klvopotek_k_means_projected_local_minimum},
\ref{thm:Klvopotek_k_means_inverted_local_minimum}.
\item We demonstrate that 
a perfect  $k$-means algorithm in the projected space is a constant factor approximation of the global optimum in the original space 
- see Theorem \ref{thm:Klvopotek_k_means_projected_global_minimum}
\item We prove that the projection preserves  several clusterability properties 
- see Theorems \ref{thm:clusterability-multPertRobust},
\ref{thm:clusterability-sigmaSeparatedness}, 
\ref{thm:clusterability-cSigmaApproximation}.
\ref{thm:clusterability-betaCentreStability}
and
\ref{thm:clusterability-weakDeletionStability}. 
\end{itemize} 
 
For $k$-means in particular we make the following claim:
\begin{theorem}\label{thm:Klvopotek_k_means_projected_cost}
Let $Q$ be a set of $m$ representatives of objects from $\mathfrak{Q}$  in an   $n$-dimensional orthogonal coordinate system $\mathcal{C}_n$.  
{ } Let $\delta \in (0,\frac12)$, $\epsilon \in (0,1)$.
and let 
\begin{equation}\label{eq:nprime_computation}
n'\ge 2\frac{-\ln \epsilon +   2 \ln( m)  } 
{\besterNenner}
\end{equation}
Let $\mathcal{C}_{n'}$ be a randomly selected (via sampling from a normal distribution)
$n'$-dimensional orthogonal coordinate system. 
Let the set $Q'$ consist of $m$ objects such that for each $i\in\mathfrak{Q}$,   
$\mathbf{x'}_i\in Q'$  is a projection of  $\mathbf{x}_i\in Q$   onto $\mathcal{C}_{n'}$.
If $\mathfrak{C}$ is a partition of $\mathfrak{Q}$,
then  
\begin{equation}\label{eq:Klvopotek_k_means_projected_cost}
(1-\delta) \mathfrak{J}(Q,\mathfrak{C}) \le \frac{n}{n'}
\mathfrak{J}(Q',\mathfrak{C})  \le (1+\delta)\mathfrak{J}(Q,\mathfrak{C})
\end{equation}
holds with probability of at least $1-\epsilon$.
\\ 
\end{theorem}

Note that the inequality  (\ref{eq:Klvopotek_k_means_projected_cost})
can be rewitten as 
\begin{equation*}
\left(1-\frac\delta{1+\delta}\right) \mathfrak{J}(Q',\mathfrak{C}) \le \frac{n'}{n}
\mathfrak{J}(Q,\mathfrak{C})  \le \left(1+\frac\delta{1-\delta}\right)\mathfrak{J}(Q',\mathfrak{C})
\end{equation*}

Furthermore 
\begin{theorem}\label{thm:Klvopotek_k_means_projected_local_minimum}
Under the assumptions and notation of Theorem \ref{thm:Klvopotek_k_means_projected_cost},
  if the partition $\mathfrak{C}^* $ constitutes a local minimum of $\mathfrak{J}(Q,\mathfrak{C})$ 
over $\mathfrak{C}$ (in the original space)
and if 
for any two clusters  
  $g=2(1-\alpha)$ times half of the distance between their centres is the gap between these clusters, where   $\alpha\in [0,1)$,
and 
\begin{equation}
\delta \le  \frac{ 1-\left(1-\frac g2\right)^2}{ \left(1-\frac g2\right)^2+(1+2p) }
\end{equation}
($p$ to be defined later by inequality (\ref{eq:pdefinition}))
then this same partition 
 is (in the projected space) also a local minimum of 
$\mathfrak{J}(Q',\mathfrak{C})$ 
over $\mathfrak{C}$, 
with probability of at least $1-\epsilon$.
\\ 
\end{theorem}

\begin{theorem}\label{thm:Klvopotek_k_means_inverted_local_minimum}
Under the assumptions and notation of Theorem \ref{thm:Klvopotek_k_means_projected_cost},
  if the clustering $\mathfrak{C'}^* $ constitutes a local minimum of $\mathfrak{J}(Q',\mathfrak{C})$ 
over $\mathfrak{C}$ (in the projected space)
and if 
for any two clusters  
  $1-\alpha$ times the distance between their centres is the gap between these clusters, where   $\alpha\in [0,1)$,
and 
\begin{equation}
\frac{ \delta}{1-\delta} \le   \frac{1-\alpha^2}{(1+2p)+\alpha^2}
\end{equation}
then the very same partition   $\mathfrak{C'}^*$
 is also (in the original space) a local minimum of 
$\mathfrak{J}(Q,\mathfrak{C})$ 
over $\mathfrak{C}$, 
with probability of at least $1-\epsilon$.
\\ 
\end{theorem}

\begin{theorem}\label{thm:Klvopotek_k_means_projected_global_minimum}
Under the assumptions and notation of Theorem \ref{thm:Klvopotek_k_means_projected_cost},
if  $\mathfrak{C_G} $ denotes the clustering reaching the global optimum in the original space, 
and $\mathfrak{C'_G} $ denotes the clustering reaching the global optimum in the projected space,
then  
\begin{equation} 
\frac{n}{n'}
 \mathfrak{J}(Q',\mathfrak{C'_G})  \le (1+\delta)\mathfrak{J}(Q,
\mathfrak{C_G})
\end{equation}
with probability of at least $1-\epsilon$.

That is the perfect $k$-means algorithm in the projected space is a constant factor approximation of $k$-means optimum in the original space. 
\end{theorem}

We postpone the proof of the theorems 
\ref{thm:Klvopotek_k_means_projected_cost}-\ref{thm:Klvopotek_k_means_projected_global_minimum} till section
\ref{sec:Klvopotek_projective_k_means},
as we need first to derive the basic theorem   \ref{thm:Klvopotek_Lemma} in section \ref{sec:derivation} which is essentially based on the results reported by Dasgupta and Gupta \cite{Dasgupta:2003}. 

Let us however stress at this point the significance of these theorems.
Earlier forms of JL lemma required sampling of the coordinates over and over again\footnote{
Though in passing a similar result is claimed in Lemma 5.3
\url{http://math.mit.edu/~bandeira/2015_18.S096_5_Johnson_Lindenstrauss.pdf}, though without an explicit proof.
}, with quite a low success rate until a mapping is found fitting the error constraints.
In our theorems, we need only one sampling in order to achieve the required success probability of selecting a suitable subspace to perform $k$-means.  
In Section \ref{sec:examples} we illustrate this advantage by some numerical simulation results, showing at the same time the impact of various parameters of Jonson-Lindenstrauss Lemma on the dimensionality of the projected space. 
In Section \ref{sec:previousWork} we recall the corresponding results of other authors. 

In Section \ref{sec:clusterability} we demonstrate an additional advantage of our version of JL lemma consisting in preservation of various clusterability criteria. 

Section \ref{sec:conclusions} contains some concluding remarks. 


\section{Derivation of the Set-Friendly \\ Johnson-Lindenstrauss Lemma} \label{sec:derivation}

Let us present the process of seeking the mapping $f$ from Theorem \ref{thm:JL_lemma}
 in a more detailed manner, so that we can then switch to our target of selecting the size of the subspace guaranteeing that the projected distances preserve their proportionality in the required range.  

Let us consider first a single vector $\mathbf{x}=(x_1,.x_2,...,x_n)$ of $n$ independent random variables drawn from the normal distribution $\mathcal{N}(0,1)$ with mean 0 and variance 1. 
Let 
$\mathbf{x'}=(x_1,.x_2,...,x_{n'})$, where $n'<n$, be its projection onto the first $n'$ coordinates.  

Dasgupta and Gupta \cite{Dasgupta:2003} in their Lemma 2.2 demonstrated that for a positive $\beta$
\begin{itemize}
\item
if $\beta<1$ then 
\begin{equation}
Pr(\|\mathbf{x}'\|^2\le \beta\frac{n'}{n} \|\mathbf{x}\|^2)\le \beta^{\frac{n'}{2}}\left(1+\frac{n'(1-\beta)}{n-n' }\right)^{\frac{n-n'}{2}}
\end{equation}
\item
if $\beta>1$ then 
\begin{equation}
Pr(\|\mathbf{x}'\|^2\ge \beta\frac{n'}{n} \|\mathbf{x}\|^2)\le \beta^{\frac{n'}{2}}\left(1+\frac{n'(1-\beta)}{n-n'}\right)^{\frac{n-n'}{2}}
\end{equation}
\end{itemize}

Now imagine we want to keep the error of squared length of $\mathbf{x}$ bounded within a range of $\pm \delta$ (relative error) upon projection, where   $\delta \in (0,1)$. Then  we get the probability 
\begin{align*}
&Pr\left((1-\delta)\|\mathbf{x} \|^2 \le \frac{n}{n'}\|\mathbf{x}'\|^2 \le (1+\delta)\|\mathbf{x} \|^2 \right)
\\&\ge 1-
  (1-\delta)^\frac{n'}{2}\left(1+\frac{n'\delta}{n-n'}\right)^{\frac{n-n'}{2}}
\\&  - (1+\delta)^{\frac{n'}{2}}\left(1-\frac{n'\delta}{n-n'}\right)^{\frac{n-n'}{2}}
\end{align*}
This implies 
  
\begin{align*}
& Pr\left((1-\delta)\|\mathbf{x} \|^2 \le \frac{n}{n'}\|\mathbf{x}'\|^2 \le (1+\delta)\|\mathbf{x} \|^2 \right)
\\&
\ge
    1-
2\max\left(  (1-\delta)^{\frac{n'}{2}}\left(1+\frac{n'\delta}{n-n'}\right)^{\frac{n-n'}{2}},
\right.\\&\left.
     (1+\delta)^{\frac{n'}{2}}\left(1-\frac{n'\delta}{n-n'}\right)^{\frac{n-n'}{2}}
\right)
\\&
=   1-
2\max_{\delta^*\in\{-\delta,+\delta\}}\left(  \left(1-\delta^*\right)^{\frac{n'}{2}}\left(1+\frac{\delta^* n'}{n-n'}\right)^{\frac{n-n'}{2}}
\right)
\end{align*}
The same holds if we scale the vector $\mathbf{x}$. 

Now if we have a sample consisting of $m$ points in space, without however a guarantee that coordinates are independent between the vectors then we want that the probability that squared distances   between all of them lie within the relative range $\pm \delta$ is higher than  
\begin{equation}\label{eq:nprime_implicit}
 1-\epsilon \le 1-{m \choose 2}\left(1- Pr\left((1-\delta)\|\mathbf{x} \|^2 \le \frac{n}{n'}\|\mathbf{x}'\|^2 \le (1+\delta)\|\mathbf{x} \|^2 \right)\right) 
\end{equation}
for some failure probability\footnote{
We speak about a success if all the projected data points lie within the range defined by formula 
(\ref{eq:JL}). Otherwise we speak about failure
(even if only one data point lies outside this range). 
} term $\epsilon \in (0,1)$. 

To achieve this, it is sufficient that the following  holds: 
$$\epsilon \ge 2{m \choose 2}\max_{\delta^*\in\{-\delta,+\delta\}}\left(  (1-\delta^*)^{\frac{n'}{2}}\left(1+\frac{\delta^* n'}{n-n'}\right)^{\frac{n-n'}{2}}
\right)$$
Taking logarithm 

\begin{align*}
 \ln \epsilon \ge &\ln(m (m-1))
\\&+ \max_{\delta^*\in\{-\delta,+\delta\}}\left( \frac{n'}{ 2} \ln (1-\delta^*)+\frac{(n-n')}{ 2} \ln\left(1+\frac{\delta^* n'}{n-n'}\right)
\right) 
\end{align*}

\begin{align*}
&\ln \epsilon -   \ln(m (m-1)) 
\\&
\ge \max_{\delta^*\in\{-\delta,+\delta\}}\left( \frac{n'}{ 2} \ln (1-\delta^*)+\frac{(n-n')}{ 2} \ln\left(1+\frac{\delta^* n'}{n-n'}\right)
\right)
\end{align*}

We know%
\footnote{Please recall at this point the Taylor expansion 
$\ln (1+x)=x-x^2/2+x^3/3-x^5/5 \dots$ 
which converges in the range (-1,1) 
and hence implies $\ln (1+x)<x$  for $x\in (-1,0)\cup (0,1)$ as we will refer to it discussing difference to  JL theorems of other authors.  
}
that 
$\ln (1+x)<x$ for $x>-1$ and $x \ne 0$, 
hence   the above holds if 
 
$$\ln \epsilon -   \ln(m (m-1)) \ge \max_{\delta^*\in\{-\delta,+\delta\}}\left(  \frac{n'}{ 2} \ln (1- \delta^*)+\frac{(n-n')}{ 2}   \frac{ \delta^* n'}{n-n'} 
\right)
 $$
$$\ln \epsilon -   \ln(m (m-1)) \ge 
\max_{\delta^*\in\{-\delta,+\delta\}}\left(  \frac{n'}{ 2} \ln (1- \delta^*)+\frac{ 1}{ 2}     (\delta^*) n'   \right)
=\frac{n'}{ 2}
\max_{\delta^*\in\{-\delta,+\delta\}}\left(   \ln (1- \delta^*)+  \delta^*    \right)
 $$
Recall that also we have $\ln (1-x)+x<0$ for $x<1$ and $x \ne 0$,
threfore 
$$\max_{\delta^*\in\{-\delta,+\delta\}}\left(2\frac{\ln \epsilon -   \ln(m (m-1))} 
{\ln (1- \delta^*)+    \delta^*}\right)
\le     n'   
 $$
So finally, realizing that  
${-\ln (1- \delta)-    \delta} \ge \besterNenner >0$, 
and that  $\ln(m(m-1))<2\ln (m)$
we get as \textit{sufficient} condition\footnote{
We substituted the denominator with a smaller positive number 
and the nominator with a larger positive number so that the fraction value increases so that a higher $n'$ will be required than actually needed. 
}

$$n'\ge 2\frac{-\ln \epsilon +   2 \ln( m)  } 
{\besterNenner}
 $$

Note that this expression does not depend on $n$ that is the number of dimensions in the projection is chosen independently of the original number of dimensions\footnote{
Though in passing a similar result is claimed in Lemma 5.3
\url{http://math.mit.edu/~bandeira/2015_18.S096_5_Johnson_Lindenstrauss.pdf}, though without an explicit proof.
They propose that 
$$n'\ge (2+r)\frac{    2 \ln( m)  } {\besterNenner}
$$ 
in order to get a failure rate below $m^{-r}$.
In fact when we substitute $\epsilon=m^{-r}$, both formulas are the same.
However, usage of $\epsilon$ alows for control of failure rate in the other theorems in this paper, while $r$ does not make this possibility obvious. Also fixing $r$ versus fixing $\epsilon$ impacts disadvantageously the growth rate of $n'$ with $m$.   
}.

So we are ready to formulate our major finding of this paper
\begin{theorem}\label{thm:Klvopotek_Lemma}
{ } Let $\delta \in (0,\frac12)$, $\epsilon \in (0,1)$.
Let $Q\subset \mathbb{R}^n$ be a set of $m$ points in an $n$-dimensional orthogonal coordinate system $\mathcal{C}_n$  
and let (as in formula~(\ref{eq:nprime_computation}))
\begin{equation*}
n'\ge 2\frac{-\ln \epsilon +   2 \ln( m)  } 
{\besterNenner}
\end{equation*}
Let $\mathcal{C}_{n'}$ be a randomly selected (via sampling from a normal distribution)
$n'$-dimensional orthogonal coordinate system. 
For each $\mathbf{v}\in Q$ let $\mathbf{v'}$ be its projection onto $\mathcal{C}_{n'}$. 
Then for \emph{all} pairs $\mathbf{u},\mathbf{v} \in Q$
\begin{equation}
(1-\delta)\|\mathbf{u}-\mathbf{v} \|^2 \le \frac{n}{n'}\|\mathbf{u'}-\mathbf{v'}\|^2 \le (1+\delta)\|\mathbf{u}-\mathbf{v} \|^2
\end{equation}
holds with probability of at least $1-\epsilon$ 
\end{theorem}

\section{Proofs of theorems \ref{thm:Klvopotek_k_means_projected_cost}-\ref{thm:Klvopotek_k_means_projected_global_minimum}}
\label{sec:Klvopotek_projective_k_means}

The permissible error $\delta$ will surely depend on the target application.
Let us consider the context of $k$-means. 
First we claim for $k$-means, that the JL Lemma applies not only to 
data points but also to cluster centres. 

\begin{lemma}\label{lem:Klvopotek_Lemma_point_cluster_centre}
{ } Let $\delta \in (0,\frac12)$, $\epsilon \in (0,1)$.
Let $Q\subset \mathbb{R}^n$ be a set of $m$ representatives of elements of $\mathfrak{Q}$ in an $n$-dimensional orthogonal coordinate system $\mathcal{C}_n$  
and let the inequality (\ref{eq:nprime_computation}) hold. 
Let $\mathcal{C}_{n'}$ be a randomly selected (via sampling from a normal distribution)
$n'$-dimensional orthogonal coordinate system. 
For each $\mathbf{x}_i\in Q$ let $\mathbf{x_i'}\in Q'$ be its projection onto $\mathcal{C}_{n'}$. 
Let $\mathfrak{C}$ be a partition of $\mathfrak{Q}$. 
Then for \emph{all} data points  $\mathbf{x_i} \in Q$
\begin{align}
(1-\delta) \|\mathbf{x_i}-\boldsymbol\mu(\mathfrak{C}(i)) \|^2 \le \frac{n}{n'} \|\mathbf{x_i'}-\boldsymbol\mu'(\mathfrak{C}(i))\|^2 \le (1+\delta) \|\mathbf{x}_i-\boldsymbol\mu(\mathfrak{C}(i)) \|^2
\label{eq:onecluster}
\Bem{\\
(1-\delta) \|\boldsymbol\mu_1-\boldsymbol\mu_2 \|^2 \le \frac{n}{n'} \|\boldsymbol\mu_1'-\boldsymbol\mu_2'\|^2 \le (1+\delta) \|\boldsymbol\mu_1-\boldsymbol\mu_2 \|^2
}
\end{align}
hold with probability of at least $1-\epsilon$,
\end{lemma}

\begin{proof}
As we know, data points under $k$-means  are assigned to clusters having the closest cluster centre. 
On the other hand the cluster centre $\boldsymbol\mu$ is the average of all the data point representatives  in the cluster.

Hence the cluster element $i$ has the 
squared distance to its cluster centre $\boldsymbol\mu(\mathfrak{C}(i))$ amounting to

$$\|\mathbf{x}_i-\boldsymbol\mu(\mathfrak{C}(i))\|^2=\frac{1}{|\mathfrak{C}(i)|}\sum_{j \in  \mathfrak{C}(i) }\|\mathbf{x_i} -\mathbf{x_j}\|^2$$ 

But according to Theorem \ref{thm:Klvopotek_Lemma}

$$
(1-\delta)\sum_{j \in  \mathfrak{C}(i) }\|\mathbf{x}_i-\mathbf{x}_j \|^2 \le \frac{n}{n'}
\sum_{j \in  \mathfrak{C}(i) }\|\mathbf{x}_i'-\mathbf{x}_j'\|^2 \le 
(1+\delta)\sum_{j \in  \mathfrak{C}(i) }\|\mathbf{x}_i-\mathbf{x}_j \|^2
$$

Hence
$$
(1-\delta) \|\mathbf{x}_i-\boldsymbol\mu(\mathfrak{C}(i)) \|^2 \le \frac{n}{n'} \|\mathbf{x}_i'-\boldsymbol\mu'(\mathfrak{C}(i))\|^2 \le (1+\delta) \|\mathbf{x}_i-\boldsymbol\mu(\mathfrak{C}(i)) \|^2
$$
Note that here $\boldsymbol\mu'(\mathfrak{C}(i))$ is not the projective image 
of $\boldsymbol\mu(\mathfrak{C}(i))$, but rather the centre of projected images of cluster elements. 

\Bem{
In a similar way we can show that for two neighbouring cluster centres we have
$$
(1-\delta) \|\boldsymbol\mu_1-\boldsymbol\mu_2 \|^2 \le \frac{n}{n'} \|\boldsymbol\mu_1'-\boldsymbol\mu_2'\|^2 \le (1+\delta) \|\boldsymbol\mu_1-\boldsymbol\mu_2 \|^2
$$
}
\end{proof}

The Lemma \ref{lem:Klvopotek_Lemma_point_cluster_centre}
permits us to prove Theorem \ref{thm:Klvopotek_k_means_projected_cost}

\begin{proof} {\bf (Theorem \ref{thm:Klvopotek_k_means_projected_cost})}
According to formula (\ref{eq:onecluster}):

$$
(1-\delta) \|\mathbf{x_i}-\boldsymbol\mu(\mathfrak{C}(i)) \|^2 \le \frac{n}{n'} \|\mathbf{x_i'}-\boldsymbol\mu'(\mathfrak{C}(i))\|^2 \le (1+\delta) \|\mathbf{x_i}-\boldsymbol\mu(\mathfrak{C}(i)) \|^2
$$
Hence  
$$
\sum_{i\in \mathfrak{Q}}
(1-\delta) \|\mathbf{x_i}-\boldsymbol\mu(\mathfrak{C}(i)) \|^2 \le 
\sum_{i\in \mathfrak{Q}}
\frac{n}{n'} \|\mathbf{x_i'}-\boldsymbol\mu'(\mathfrak{C}(i))\|^2 \le 
\sum_{i\in \mathfrak{Q}}
(1+\delta) \|\mathbf{x_i}-\boldsymbol\mu(\mathfrak{C}(i)) \|^2
$$

$$
(1-\delta)\sum_{i\in \mathfrak{Q}}
 \|\mathbf{x_i}-\boldsymbol\mu(\mathfrak{C}(i)) \|^2 \le 
\sum_{i\in \mathfrak{Q}}
\frac{n}{n'} \|\mathbf{x_i'}-\boldsymbol\mu'(\mathfrak{C}(i))\|^2 \le 
(1+\delta)\sum_{i\in \mathfrak{Q}}
 \|\mathbf{x_i}-\boldsymbol\mu(\mathfrak{C}(i)) \|^2
$$

Based on defining equations (\ref{eq:defJ}) and (\ref{eq:defJprim})
we get the formula (\ref{eq:Klvopotek_k_means_projected_cost})
\begin{equation*}
(1-\delta) \mathfrak{J}(Q,\mathfrak{C}) \le \frac{n}{n'}
\mathfrak{J}(Q',\mathfrak{C})  \le (1+\delta)\mathfrak{J}(Q,\mathfrak{C})
\end{equation*}
\end{proof}

Let us now investigate the distance between centres of two clusters, say $C_1, C_2$.
Let their cardinalities amount to $m_1,m_2$ respectively. 
Denote $C_{12}=C_1\cup C_2$.
Consequently $m_{12}=|C_{12}|=m_1+m_2$. 
For a set $C_j$ let $VAR(C_j)=\frac1 {|C_j|}\sum_{i\in C_j} 
\|\mathbf{x_i}-\boldsymbol\mu(C_j)\|^2$ and 
 $VAR'(C_j)=\frac1 {|C_j|}\sum_{i\in C_j} 
\|\mathbf{x'_i}-\boldsymbol\mu'(C_j)\|^2$.

Therefore 
 $$VAR(C_{12})=\frac1 {|C_{12}|}\sum_{i\in C_{12}} 
\|\mathbf{x_i}-\boldsymbol\mu(C_{12})\|^2$$
$$=\frac1 {|C_{12}|}\left(
\left(
\sum_{i\in C_{1}} 
\|\mathbf{x_i}-\boldsymbol\mu(C_{12})\|^2
\right)
+\left(
\sum_{i\in C_{2}} 
\|\mathbf{x_i}-\boldsymbol\mu(C_{12})\|^2
\right)
\right)
$$
By inserting a zero 
$$=\frac1 {|C_{12}|}\left(
\left(
\sum_{i\in C_{1}} 
\|\mathbf{x_i}-\boldsymbol\mu(C_{1})+\boldsymbol\mu(C_{1})-\boldsymbol\mu(C_{12})\|^2
\right)
+\left(
\sum_{i\in C_{2}} 
\|\mathbf{x_i}-\boldsymbol\mu(C_{12})\|^2
\right)
\right)
$$

$$=\frac1 {|C_{12}|}\left(
\left(
\sum_{i\in C_{1}} \left(
(\mathbf{x_i}-\boldsymbol\mu(C_{1}))^2
+(\boldsymbol\mu(C_{1})-\boldsymbol\mu(C_{12}))^2
+2(\mathbf{x_i}-\boldsymbol\mu(C_{1}))(\boldsymbol\mu(C_{1})-\boldsymbol\mu(C_{12})) 
\right)\right)
\right. $$ $$ \left. 
+\left(
\sum_{i\in C_{2}} 
\|\mathbf{x_i}-  \boldsymbol\mu(C_{12})\|^2
\right)
\right)
$$

$$=\frac1 {|C_{12}|}\left(
\left(
 \left(\sum_{i\in C_{1}}(\mathbf{x_i}-\boldsymbol\mu(C_{1}))^2\right)
+\left(\sum_{i\in C_{1}}(\boldsymbol\mu(C_{1})-\boldsymbol\mu(C_{12}))^2\right)
\right. \right. $$ $$ \left. \left. 
+2(\sum_{i\in C_{1}}\mathbf{x_i}-\sum_{i\in C_{1}}\boldsymbol\mu(C_{1}))(\boldsymbol\mu(C_{1})-\boldsymbol\mu(C_{12})) 
\right)\right)
+\left(
\sum_{i\in C_{2}} 
\|\mathbf{x_i}-\boldsymbol\mu(C_{2})+\boldsymbol\mu(C_{2})-\boldsymbol\mu(C_{12})\|^2
\right)
$$

$$=\frac1 {|C_{12}|}\left(
\left(
 \left(\sum_{i\in C_{1}}(\mathbf{x_i}-\boldsymbol\mu(C_{1}))^2\right)
+| C_{1}|(\boldsymbol\mu(C_{1})-\boldsymbol\mu(C_{12}))^2
\right. \right. $$ $$ \left. \left. 
+2(| C_{1}|\boldsymbol\mu(C_{1})-|C_{1}|\boldsymbol\mu(C_{1}))(\boldsymbol\mu(C_{1})-\boldsymbol\mu(C_{12})) 
\right)
+\left(
\sum_{i\in C_{2}} 
\|\mathbf{x_i}-\boldsymbol\mu(C_{12})\|^2
\right)
\right)
$$

$$=\frac1 {|C_{12}|}\left(
\left(VAR(C_1) |C_1|
+| C_{1}|(\boldsymbol\mu(C_{1})-\boldsymbol\mu(C_{12}))^2
\right) 
+\left(
\sum_{i\in C_{2}} 
\|\mathbf{x_i}- \boldsymbol\mu(C_{12})\|^2
\right)
\right)
$$
Via the same reasonig we get: 
$$=\frac1 {|C_{12}|}\left(
\left(VAR(C_1) |C_1|
+| C_{1}|(\boldsymbol\mu(C_{1})-\boldsymbol\mu(C_{12}))^2 
\right) 
\right. $$ $$ \left. 
+\left(VAR(C_2) |C_2|
+| C_{2}|(\boldsymbol\mu(C_{2})-\boldsymbol\mu(C_{12}))^2 
\right) 
\right)
$$
$$=\frac1 {|C_{12}|}\left(
VAR(C_1) |C_1|+VAR(C_2) |C_2|
\right. $$ $$ \left. 
+| C_{1}|(\boldsymbol\mu(C_{1})-\boldsymbol\mu(C_{12}))^2 
+| C_{2}|(\boldsymbol\mu(C_{2})-\boldsymbol\mu(C_{12}))^2
\right)
$$
As  
Apparently
$\boldsymbol\mu(C_{12})=\frac{1}{|C_{12}|}\sum_{i\in C_{12}} \mathbf{x_i}
$
$ =\frac{1}{|C_{12}|}\left((\sum_{i\in C_{1}} \mathbf{x_i})+(\sum_{i\in C_{2}} \mathbf{x_i})\right)
$
$ =\frac{1}{|C_{12}|}\left( |C_{1}| \boldsymbol\mu(C_{1})
+|C_{1}| \boldsymbol\mu(C_{1})\right)
$
that is 
$\boldsymbol\mu(C_{12})=
\frac{|C_1|}{|C_{12}|}\boldsymbol\mu(C_{1})
+\frac{|C_2|}{|C_{12}|}\boldsymbol\mu(C_{2}$, we get 
$$=\frac1 {|C_{12}|}\left(
VAR(C_1) |C_1|+VAR(C_2) |C_2|
+| C_{1}|\left(\boldsymbol\mu(C_{1})-\frac{|C_1|}{|C_{12}|}\boldsymbol\mu(C_{1})
-\frac{|C_2|}{|C_{12}|}\boldsymbol\mu(C_{2}\right)^2 
\right. $$ $$ \left. 
+| C_{2}|\left(\boldsymbol\mu(C_{2})-\frac{|C_1|}{|C_{12}|}\boldsymbol\mu(C_{1})
-\frac{|C_2|}{|C_{12}|}\boldsymbol\mu(C_{2})\right)^2
\right)
$$
$$=\frac1 {|C_{12}|}\left(
VAR(C_1) |C_1|+VAR(C_2) |C_2|
+| C_{1}|\left( \frac{|C_2|}{|C_{12}|}\boldsymbol\mu(C_{1})
-\frac{|C_2|}{|C_{12}|}\boldsymbol\mu(C_{2})
\right)^2 
\right. $$ $$ \left. 
+| C_{2}|\left( -\frac{|C_1|}{|C_{12}|}\boldsymbol\mu(C_{1})
+\frac{|C_1|}{|C_{12}|}\boldsymbol\mu(C_{2})\right)^2 
\right)
$$
 
$$=\frac1 {|C_{12}|}\left(
VAR(C_1) |C_1|+VAR(C_2) |C_2|
+
\frac{|C_1||C_2|^2+|C_1|^2|C_2|}{|C_{12}|^2}  \left( \boldsymbol\mu(C_{1})
- \boldsymbol\mu(C_{2})\right)^2
\right)
$$

hence 
$$VAR(C_{12})==\frac1 {|C_{12}|}\left(
VAR(C_1) |C_1|+VAR(C_2) |C_2|
+
\frac{|C_1||C_2| }{|C_{12}| }  \left( \boldsymbol\mu(C_{1})
- \boldsymbol\mu(C_{2})
\right)^2
\right)
$$


This leads immediately to 

$$
  VAR(C_{12})\cdot m_{12}= VAR(C_1)\cdot m_1+VAR(C_2)\cdot m_2+m_1\cdot m_2/m_{12}\cdot \|\boldsymbol\mu(C_1)-\boldsymbol\mu(C_2)\|^2 $$
which implies
$$
  VAR(C_{12})\cdot \frac{m_{12}^2}{ m_1 \cdot m_2}= VAR(C_1)\cdot \frac{m_{12}}{m_2}+VAR(C_2)\cdot \frac{m_{12}}{m_1}+ \|\boldsymbol\mu(C_1)-\boldsymbol\mu(C_2)\|^2 $$

According to Lemma \ref{lem:Klvopotek_Lemma_point_cluster_centre}, applied to the set $C_{12}$ as a cluster,

$$
(1-\delta)
\left(  VAR(C_1)\cdot m_{12}/m_2+VAR(C_2)\cdot m_{12}/m_1+ \|\boldsymbol\mu(C_1)-\boldsymbol\mu(C_2)\|^2 \right)
$$ $$\le 
\frac{n}{n'} 
\left(  VAR'(C_1)\cdot m_{12}/m_2+VAR'(C_2)\cdot m_{12}/m_1+ \|\boldsymbol\mu'(C_1)-\boldsymbol\mu'(C_2)\|^2 \right)
$$ $$\le 
(1+\delta)
\left(  VAR(C_1)\cdot m_{12}/m_2+VAR(C_2)\cdot m_{12}/m_1+ \|\boldsymbol\mu(C_1)-\boldsymbol\mu(C_2)\|^2 \right)
 $$
and with respect to $C_1,C_2$ combined 
$$
(1-\delta)
\left(  VAR(C_1)\cdot m_{12}/m_2+VAR(C_2)\cdot m_{12}/m_1  \right)
$$ $$\le 
\frac{n}{n'} 
\left(  VAR'(C_1)\cdot m_{12}/m_2+VAR'(C_2)\cdot m_{12}/m_1  \right)
$$ $$\le 
(1+\delta)
\left(  VAR(C_1)\cdot m_{12}/m_2+VAR(C_2)\cdot m_{12}/m_1 \right)
 $$
These two last equations mean that

$$
-2\delta
\left(  VAR(C_1)\cdot m_{12}/m_2+VAR(C_2)\cdot m_{12}/m_1\right)+ (1-\delta)\|\boldsymbol\mu(C_1)-\boldsymbol\mu(C_2)\|^2 
$$ $$\le 
\frac{n}{n'} 
\left(   \|\boldsymbol\mu'(C_1)-\boldsymbol\mu'(C_2)\|^2 \right)
$$ $$\le 
2\delta
\left(  VAR(C_1)\cdot m_{12}/m_2+VAR(C_2)\cdot m_{12}/m_1\right)+ (1+\delta)\|\boldsymbol\mu(C_1)-\boldsymbol\mu(C_2)\|^2 
 $$
Let us assume that the quotient 
\begin{equation}\label{eq:pdefinition}
\frac{VAR(C_1)\cdot m_{12}/m_2+VAR(C_2)\cdot m_{12}/m_1}{\|\boldsymbol\mu(C_1)-\boldsymbol\mu(C_2)\|^2} \le p
\end{equation}
where $p$ is some positive number. 
So we have in effect 
$$
  (1-\delta(1+2p))\|\boldsymbol\mu(C_1)-\boldsymbol\mu(C_2)\|^2 
\le 
\frac{n}{n'} 
\left(   \|\boldsymbol\mu'(C_1)-\boldsymbol\mu'(C_2)\|^2 \right)
\le 
  (1+\delta(1+2p))\|\boldsymbol\mu(C_1)-\boldsymbol\mu(C_2)\|^2 
 $$

Under balanced ball-shaped clusters $p$ does not exceed 1. 
So we have shown the lemma
\begin{lemma}  \label{lem:Klvopotek_Lemma_two_cluster_centres}
Under the assumptions of preceding lemmas 
for any two clusters $C_1,C_2$ 
\begin{equation}
  (1-\delta(1+2p))\|\boldsymbol\mu(C_1)-\boldsymbol\mu(C_2)\|^2 
\le 
\frac{n}{n'} 
\left(   \|\boldsymbol\mu'(C_1)-\boldsymbol\mu'(C_2)\|^2 \right)
\le 
  (1+\delta(1+2p))\|\boldsymbol\mu(C_1)-\boldsymbol\mu(C_2)\|^2 
\end{equation}
where $p$ depends on degree of balance between clusters and cluster shape,
holds with probability at least $1-\epsilon$. 
\end{lemma}

Now let us consider the choice of $\delta$ in such a way that with high probability no data point will be classified into some other cluster. 
We claim the following
\begin{lemma}\label{lem:Klvopotek_Lemma_no__cluster_change}
Consider two   clusters $C_1,C_2$.
{ } Let $\delta \in (0,\frac12)$, $\epsilon \in (0,1)$.
Let $Q\subset \mathbb{R}^n$ be a set of $m$ points in an $n$-dimensional orthogonal coordinate system $\mathcal{C}_n$  
and let the inequality (\ref{eq:nprime_computation}) hold. 
Let $\mathcal{C}_{n'}$ be a randomly selected (via sampling from a normal distribution)
$n'$-dimensional orthogonal coordinate system. 
For each $\mathbf{x}_i\in Q$ let $\mathbf{v'}$ be its projection onto $\mathcal{C}_{n'}$. 
For   two clusters $C_1,C_2$, obtained via  $k$-means, in the original space let 
  $\boldsymbol\mu_1,  \boldsymbol\mu_2$ be their centres   
and
$\boldsymbol\mu_1',  \boldsymbol\mu_2'$ be centres to the correspondings sets of projected cluster members. 
Furthermore let $d$ be the distance of the first cluster centre to the common border of both clusters 
and let  the closest point of the first cluster to this border be at the distance  of 
  $\alpha d$ from its cluster centre as projected on the line connecting both cluster centres, where   $\alpha\in (0,1)$.
\\
Then all 
projected points  of the first cluster are (each) closer 
to the centre of the set of projected points of the first 
than to 
 the centre of the set of projected points of the second 
  if 
\begin{equation}
\delta \le  \frac{ 1-\left(1-\frac g2\right)^2}{ \left(1-\frac g2\right)^2+(1+2p) }=\frac{1-\alpha^2}{(1+2p)+\alpha^2}
\end{equation}
where $g=2(1-\alpha)$, 
 with  probability of at least $1-\epsilon$.
\end{lemma}

\begin{proof}
Consider    a data point $\mathbf{x}$ "close" to the border between the two neighbouring clusters, 
on the line connecting the cluster centres, belonging to the first cluster, 
at a distance $\alpha d$ from its cluster centre, where $d$ is the distance of the first cluster centre to the border   and $\alpha\in (0,1)$.
The squared distance between cluster centres, under projection, can be "reduced" by the factor $1-\delta$, (beside the factor $\frac{n}{n'}$ which is common to all the points) 
whereas the squared distance of $\mathbf{x}$ to its cluster centre may be "increased" by the factor $1+\delta$. 
This implies a relationship between the factor $\alpha$ and the error $\delta$.

If $\mathbf{x'}$ should not cross the border between the clusters, 
the following needs to hold:

\begin{equation}\label{eq:ontheonesideofclusterborder}
\|\mathbf{x}'-\boldsymbol\mu'_1\| \le \frac12\|\boldsymbol\mu'_2-\boldsymbol\mu'_1\|  
\end{equation}
which implies:
$$ \frac{n}{n'} \|\mathbf{x}'-\boldsymbol\mu'_1\|^2 \le  \frac{n}{n'} \frac{1}{4}\|\boldsymbol\mu'_2-\boldsymbol\mu'_1\|^2 $$
As  (see Lemma \ref{lem:Klvopotek_Lemma_point_cluster_centre})
$$ \frac{n}{n'} \|\mathbf{x}'-\boldsymbol\mu'_1\|^2 \le (1+\delta) \|\mathbf{x}-\boldsymbol\mu_1\|^2=(1+\delta)(\alpha d)^2$$ 
and  (see Lemma \ref{lem:Klvopotek_Lemma_two_cluster_centres})
$$\frac{n}{n'} \frac{1}{4}\|\boldsymbol\mu'_2-\boldsymbol\mu'_1\|^2 \ge (1-\delta(1+2p)) \frac{1}{4}\|\boldsymbol\mu_2-\boldsymbol\mu_1\|^2=(1-\delta(1+2p)) d^2$$
we know that, for inequality (\ref{eq:ontheonesideofclusterborder}) to hold, 
 it is sufficient that: 
$$(1+\delta)(\alpha d)^2\le (1-\delta(1+2p)) d^2$$
that is 
$$\alpha\le \sqrt{\frac{1-\delta(1+2p)}{1+\delta}}$$
But $2(1-\alpha)d$ or  $2(1-\alpha)$ can be viewed as absolute or relative gap between clusters.
So if we expect a relative gap $g=2(1-\alpha)$ between clusters, we have to choose $\delta$ in such a way that
$$1-\frac g2\le \sqrt{\frac{1-\delta(1+2p)}{1+\delta}}$$
Therefore
\begin{equation}\label{eq:gapanderror}
\delta \le  \frac{ 1-\left(1-\frac g2\right)^2}{ \left(1-\frac g2\right)^2+(1+2p) }
\end{equation}
\end{proof}

So we see that the decision on the permitted error depends on the size of the gap between clusters that we hope to observe. 

The Lemma \ref{lem:Klvopotek_Lemma_no__cluster_change} 
allows us   to prove Theorem \ref{thm:Klvopotek_k_means_projected_local_minimum} 
in a straight forward manner. 

\begin{proof} {\bf (Theorem \ref{thm:Klvopotek_k_means_projected_local_minimum})}
Observe that in this theorem we impose the condition of this lemma on each cluster. 
So all projected points are closer to their set centres than to any other centre.
So the $k$-means algorithm would get stuck at this clustering and hence we get at a local minimum.
\end{proof}

\begin{lemma}\label{lem:Klvopotek_Lemma_no__cluster_change_inverted}
{ } Let $\delta \in (0,\frac12)$, $\epsilon \in (0,1)$.
Let $Q\subset \mathbb{R}^n$ be a set of $m$ points in an $n$-dimensional orthogonal coordinate system $\mathcal{C}_n$  
and let the inequility (\ref{eq:nprime_computation}) hold. 
Let $\mathcal{C}_{n'}$ be a randomly selected (via sampling from a normal distribution)
$n'$-dimensional orthogonal coordinate system. 
For each $\mathbf{x}_i\in Q$ let $\mathbf{v'}$ be its projection onto $\mathcal{C}_{n'}$. 
For any two $k$-means clusters $C_1,C_2$ in the projected space  let 
  $\boldsymbol\mu_1',  \boldsymbol\mu_2'$ be their centres in the   
projected space
and
$\boldsymbol\mu_1,  \boldsymbol\mu_2$ be centres to the corresponding sets of  cluster members in the original space. 
Furthermore let $d$ be the distance of the first cluster centre to the common border of both clusters in the projected space
and let  the closest point of the first cluster to this border in that space be at the distances of 
  $\alpha d$ from its cluster centre%
, where   $\alpha\in [0,1)$.
\\
Then all 
 points  of the first cluster in the original space are (each) closer 
to the centre of the set of   points of the first  
than to 
 the centre of the set of points of the second cluster in the original space
  if 
\begin{equation}
 \delta \le  \frac{ 1-\left(1-\frac \ggg2\right)^2}{ \left(1-\frac \ggg2\right)^2+(1+2p) }=\frac{1-\alpha^2}{(1+2p)+\alpha^2}
\end{equation}
 with  probability of at least $1-\epsilon$.
\end{lemma}
\begin{proof} 
Consider    a data point $\mathbf{x'}$ "close" to the border between the two neighbouring clusters in the projected space, 
on the line connecting the cluster centres, belonging to the first cluster, 
at a distance $\alpha d$ from its cluster centre, where $d$ is the distance of the first cluster centre to the border   and $\alpha\in (0,1)$.
The squared distance between cluster centres, in original space, can be "reduced" by the factor $(1+\delta)^{-1}$  (beside the factor $\frac{n'}{n}$ which is common to all the points), 
whereas the squared distance of $\mathbf{x}$ to its cluster centre may be "increased" by the factor  $(1-\delta)^{-1}$ . 
This implies a relationship between the factor $\alpha$ and the error $\delta$.

If $\mathbf{x}$ (in the original space) should not  cross the border between the clusters, 
the following needs to hold:

\begin{equation}\label{eq:ontheonesideofclusterborderInv}
\|\mathbf{x}-\boldsymbol\mu_1\| \le \frac12\|\boldsymbol\mu_2-\boldsymbol\mu_1\|  
\end{equation}
which implies:
$$ \frac{n'}{n} \|\mathbf{x}-\boldsymbol\mu_1\|^2 \le  \frac{n'}{n} \frac{1}{4}\|\boldsymbol\mu_2-\boldsymbol\mu_1\|^2 $$
As (see Lemma \ref{lem:Klvopotek_Lemma_point_cluster_centre})
$$ \frac{n'}{n} \|\mathbf{x}-\boldsymbol\mu_1\|^2 \le (1- \delta)^{-1} \|\mathbf{x'}-\boldsymbol\mu'_1\|^2=(1-\delta)^{-1}(\alpha d)^2$$ 
and  (see Lemma \ref{lem:Klvopotek_Lemma_two_cluster_centres})
$$\frac{n'}{n} \frac{1}{4}\|\boldsymbol\mu_2-\boldsymbol\mu_1\|^2 \ge (1+\delta(1+2p))^{-1} \frac{1}{4}\|\boldsymbol\mu'_2-\boldsymbol\mu'_1\|^2=(1+\delta(1+2p))^{-1} d^2$$
Thus, we know that, for inequality (\ref{eq:ontheonesideofclusterborderInv}) to hold, 
 it is sufficient that: 
$$(1-\delta)^{-1}(\alpha d)^2\le (1+\delta(1+2p))^{-1} d^2$$
that is 
$$\alpha\le \sqrt{\frac{1-\delta}{1+\delta(1+2p)}}$$
But $2(1-\alpha)d$ or  $2(1-\alpha)$ can be viewed as absolute or relative gap between clusters.
So if we want to have a relative gap $g=2(1-\alpha)$ between clusters, we have to choose $\delta$ in such a way that
$$1-\frac g2\le \sqrt{\frac{1-\delta}{1+\delta(1+2p)}}$$
Therefore
\begin{equation}\label{eq:gapanderrorZwei}
\delta \le  \frac{ 1-\left(1-\frac g2\right)^2}{ \left(1-\frac g2\right)^2+(1+2p) }
\end{equation}
\end{proof}

The Lemma \ref{lem:Klvopotek_Lemma_no__cluster_change_inverted} 
allows us   to prove Theorem \ref{thm:Klvopotek_k_means_inverted_local_minimum} 
in a straight forward manner. 

\begin{proof} {\bf (Theorem \ref{thm:Klvopotek_k_means_inverted_local_minimum})}
Observe that in this theorem we impose the condition of this lemma on each cluster. 
So all original space points are closer to their set centres  than to any other centre.
So the $k$-means algorithm would get stuck at this clustering and hence we get at a local minimum.
\end{proof}

Having these results, we can go over to the proof of the Theorem
\ref{thm:Klvopotek_k_means_projected_global_minimum}.
\begin{proof} {\bf (Theorem \ref{thm:Klvopotek_k_means_projected_global_minimum})}
 Let $\mathfrak{C_G} $ denote the clustering reaching the global optimum in the original space. 
Let $\mathfrak{C'_G} $ denote the clustering reaching the global optimum in the projected space. 
From the Theorem \ref{thm:Klvopotek_k_means_projected_cost} we have that 

$$(1-\delta) \mathfrak{J}(Q,\mathfrak{C_G}) \le \frac{n}{n'}
\mathfrak{J}(Q',\mathfrak{C_G})  \le (1+\delta)\mathfrak{J}(Q,
\mathfrak{C_G})
$$

On the other hand 
$$(1+\delta)^{-1} \mathfrak{J}(Q',\mathfrak{C'_G}) \le \frac{n'}{n}
\mathfrak{J}(Q,\mathfrak{C'_G})  \le (1-\delta)^{-1}\mathfrak{J}(Q',
\mathfrak{C'_G})
$$

As $\mathfrak{C'_G} $  is the global minimum in the projected space, hence 
$$ \mathfrak{J}(Q',\mathfrak{C'_G})\le 
\mathfrak{J}(Q',\mathfrak{C_G})$$
So 
$$\frac{n}{n'}
 \mathfrak{J}(Q',\mathfrak{C'_G}) \le \frac{n}{n'}
\mathfrak{J}(Q',\mathfrak{C_G})  \le (1+\delta)\mathfrak{J}(Q,
\mathfrak{C_G})
$$
So 
$$\frac{n}{n'}
 \mathfrak{J}(Q',\mathfrak{C'_G})  \le (1+\delta)\mathfrak{J}(Q,
\mathfrak{C_G})
$$

Note that analogously,
 $\mathfrak{C_G} $  is the global minimum in the original space, hence 
$$ \mathfrak{J}(Q,\mathfrak{C_G})\le 
\mathfrak{J}(Q,\mathfrak{C'_G})$$

\begin{equation}\label{eq:global2}
 \frac{n'}{n}\mathfrak{J}(Q,\mathfrak{C_G})\le 
 \frac{n'}{n}\mathfrak{J}(Q,\mathfrak{C'_G}) \le 
 (1-\delta)^{-1}\mathfrak{J}(Q',\mathfrak{C'_G})
\end{equation}

\end{proof}

\section{Clusterability and the dimensionality reduction}\label{sec:clusterability}

In the literature a number of notions of so-called clusterability have been introduced. 
 Under these notions of clusterability algorithms have been developed clustering the data nearly optimally in polynomial times, when some constraints are matched by the  clusterability parameters. 

It seems therefore worth to have a look at the issue if the aforementioned projection technique would affect the clusterability property of the data sets. 

Let us consider, as representatives,    the following notions of clusterability, present in the literature:
\begin{itemize}
\item \emph{Perturbation Robustness} meaning that   small
perturbations of distances / positions in space of set elements  do not result in a change of the optimal clustering for that data 
set. Two brands may be distinguished: additive \cite{Ackerman:2009} and multiplicative ones \cite{Bilu:2012} (the limit of perturbation is upper-bounded either by an absolute value or by a coefficient).  
\\
The $s$-Multiplicative Perturbation Robustness ($0<s<1)$ 
holds for a data set  
 with  $d_1$ being its distance function
if the following holds.
Let $\mathfrak{C}$ be an optimal clustering 
of data points for this distance. 
Let $d_2$ be any distance function over the same set of points such that for any two points $\mathbf{u}, \mathbf{v}$, 
$s\cdot d_1(\mathbf{u}, \mathbf{v})<d_2(\mathbf{u}, \mathbf{v})<\frac1s \cdot d_1(\mathbf{u}, \mathbf{v})$. 
Then the same clustering  $\mathfrak{C}$  is optimal under the distance function $d_2$. 
\\ 
The $s$-Additive Perturbation Robustness ($0<s<1)$ 
holds for a data set  
 with  $d_1$ being its  distance function
if the following holds.
Let $\mathfrak{C}$ be an optimal clustering 
of data points for this distance. 
Let $d_2$ be any distance function over the same set of points such that for any two points $\mathbf{u}, \mathbf{v}$, 
$ d_1(\mathbf{u}, \mathbf{v})-s<d_2(\mathbf{u}, \mathbf{v})<  \cdot d_1(\mathbf{u}, \mathbf{v})+s$. 
Then the same clustering  $\mathfrak{C}$  is optimal under the distance function $d_2$. 
\\Subsequently we are interested only in the multiplicative version.
\item \emph{$\sigma$-Separatedness} \cite{Ostrovsky:2013} meaning that the cost 
$ \mathfrak{J}(Q, \mathfrak{C}_k)$ 
of optimal  clustering $\mathfrak{C}_k$ of the data set $Q$ 
into $k$ clusters is less than $\sigma^2$ ($0<\sigma<1$) times the cost
$ \mathfrak{J}(Q, \mathfrak{C}_{k-1})$ 
 of optimal clustering $\mathfrak{C}_{k-1}$ into $k-1$ clusters
$ \mathfrak{J}(Q, \mathfrak{C}_{k-1}) <  \sigma^2 \mathfrak{J}(Q, \mathfrak{C}_{k-1})$   
\item \emph{
$(c, \sigma)$-Approximation-Stability} 
\cite{Balcan:2009} meaning that if the cost function values of two partitions 
$\mathfrak{C}_a,\mathfrak{C}_b$ 
differ by at most the factor $c>1$ (that is $c\cdot\mathfrak{J}(Q,\mathfrak{C}_a)\ge\mathfrak{J}(Q,\mathfrak{C}_b)$ and $c\cdot\mathfrak{J}(Q,\mathfrak{C}_b)\ge\mathfrak{J}(Q,\mathfrak{C}_a)$), then the distance (in some space) between the partitions is at most $\sigma$ ($d(\mathfrak{C}_a,\mathfrak{C}_b)<\sigma$ for some distance function $d$ between partiitions).
As Ben-David  \cite{Ben-David:2015} recalls, this implies the uniqueness of optimal solution.
\item \emph{
$\beta$-Centre Stability} \cite{Awasthi:2012} meaning, for any centric clustering, that the distance of an element to its cluster centre is $\beta>1$ times smaller than the distance to any other cluster centre under optimal clustering. 
\item \emph{$(1+\beta)$ Weak Deletion Stability} \cite{Awasthi:2010} ($\beta>0$) meaning that given an optimal cost function value $OPT$ for $k$ centric clusters,  the cost function of a clustering obtained by deleting one of the cluster centres and assigning elements of that cluster to one of the remaining clusters should be bigger than  $(1+\beta)\cdot OPT$.
\end{itemize}

Let us first have a look at the $\sigma$-Separatedness. 
Let $\mathfrak{C_{G,k}}$ denote 
an optimal clustering into $k$ clusters  in the original space 
and 
$\mathfrak{C'_{G,k}}$ in the  projected space.
From properties of $k$-means we know that 
 $\mathfrak{J}(Q,\mathfrak{C_{G,k}})\le  \mathfrak{J}(Q,\mathfrak{C_{G,k-1}})$
and
 $\mathfrak{J}(Q',\mathfrak{C'_{G,k}})\le  \mathfrak{J}(Q',\mathfrak{C'_{G,k-1}})$.
 From theorem \ref{thm:Klvopotek_k_means_projected_global_minimum} we know that 
\begin{equation*} 
\frac{n}{n'}
 \mathfrak{J}(Q',\mathfrak{C'_{G,k}})  \le (1+\delta)\mathfrak{J}(Q,
\mathfrak{C_{G,k}})
\end{equation*}
and
\begin{equation*}  
\mathfrak{J}(Q,\mathfrak{C_{G,k-1}}) \le 
 \frac{n}{n'} (1-\delta)^{-1}\mathfrak{J}(Q',\mathfrak{C'_{G,k-1}})
\end{equation*}

$\sigma$-Separatedness implies that 
$$\sigma^2\ge 
\frac{\mathfrak{J}(Q,\mathfrak{C_{G,k}}) }
     {\mathfrak{J}(Q,\mathfrak{C_{G,k-1}}) }
\ge
\frac{ \frac{n}{n'} (1+\delta)^{-1} \mathfrak{J}(Q',\mathfrak{C_{G,k}}) }
     {\mathfrak{J}(Q,\mathfrak{C_{G',k-1}})  }
$$ $$\ge
\frac{\frac{n}{n'} (1+\delta)^{-1} \mathfrak{J}(Q',\mathfrak{C'_{G,k}})}
     {\frac{n}{n'} (1-\delta)^{-1} \mathfrak{J}(Q',\mathfrak{C'_{G,k-1}})}
=
\frac{  (1-\delta)  \mathfrak{J}(Q',\mathfrak{C'_{G,k}})}
     {  (1+\delta)  \mathfrak{J}(Q',\mathfrak{C'_{G,k-1}})}
$$ 

This implies 
$$\sigma^2\frac{1+\delta}{1-\delta}\ge 
\frac{  \mathfrak{J}(Q',\mathfrak{C'_{G,k}})}
     {  \mathfrak{J}(Q',\mathfrak{C'_{G,k-1}})}
$$ 

So we claim 
\begin{theorem}\label{thm:clusterability-sigmaSeparatedness}
Under the assumptions and notation of Theorem \ref{thm:Klvopotek_k_means_projected_cost},
If the data set $Q$ has the property of $\sigma$-Separatedness in the original space, 
then 
with probability at least $1-\epsilon$
it has the property of 
$\sigma\sqrt{\frac{1+\delta}{1-\delta}}$-Separatedness in the projected space. 
\end{theorem}

The fact that this Separatedness increases is of course a defficiency, because clustring algorithms require as low Separatedness as possible  (because the clusters are then better separated).

Let us turn to the $(c,\sigma)$-Approximation-Stability.
We can reformulate it as follows:
if
 the distance (in some space) between the partitions is more than  $\sigma$ then 
 the cost function values of two partitions differ by at least   the factor $c>1$.
Consider now two partitions $C_1, C_2$, with distance over $\sigma$ in some abstract partition space, not related to the embedding spaces. 
Then in the original space the following must hold. 
$$\mathfrak{J}(Q,\mathfrak{C_1}) \ge c \cdot \mathfrak{J}(Q,
\mathfrak{C_2})$$
Under the projection we get 
$$(1-\delta)^{-1}\frac{n}{n'}\mathfrak{J}(Q',\mathfrak{C_1}) \ge c \cdot (1+\delta)^{-1} \frac{n}{n'} \mathfrak{J}(Q',
\mathfrak{C_2})$$
$$\mathfrak{J}(Q',\mathfrak{C_1}) \ge c \cdot \frac{1-\delta}{1+\delta}  \mathfrak{J}(Q',
\mathfrak{C_2})$$

This result means that
\begin{theorem}\label{thm:clusterability-cSigmaApproximation}
Under the assumptions and notation of Theorem \ref{thm:Klvopotek_k_means_projected_cost},
if the data set $Q$ has the property of $(c, \sigma)$-Approximation-
Stability in the original space, 
then 
with probability at least $1-\epsilon$
it has the property of 
$(c \cdot \frac{1-\delta}{1+\delta} ,\sigma$-Approximation Stability property in the projected space. 
\end{theorem}

Let us now consider $s$-Multiplicative Perturbation Stability.
We claim that 

\begin{lemma}\label{lem:Klvopotek_double_perturbation}
If the data set $Q$ has the property of $\sqrt{s}$-Multiplicative Perturbation Robustness under the distance $\sqrt{d_1}$,
and the set $Q_p$ is its perturbation with distance $\sqrt{d_2}$ such that 
$\nu  d_1\le d_2\le\frac 1\nu  d_1$,
and $s=\nu \cdot s_p$, where $0<\nu ,s_p<1$,
then 
set $Q_p$ has the property of $\sqrt{s_p}$-Multiplicative Perturbation Robustness
\end{lemma}
\begin{proof}
Apparently $Q_p$  is a perturbation
of $Q$ such that both share same optimal clustering.
Let $Q_q$ be a perturbation of $Q_p$, with distance $\sqrt{d_3}$, such that 
$s_p d_2\le d_3 \le \frac{1}{s_p} d_2$.
Then 
$sd_1=s_p \nu  d_1 \le s_p d_2\le d_3 \le \frac{1}{s_p} d_2 \le \frac{1}{s_p\nu } d_1 =\frac{1}{s} d_1$
that is $Q_q$ is a perturbation of $Q$ such that both share same optimal clustering. 
So $Q_p$ and $Q_q$ share common optimal clustering, hence 
$Q_p$ has 
the property of $\sqrt{s_p}$-Multiplicative Perturbation Robustness
\end{proof}

We claim that 

\begin{lemma}\label{lem:Klvopotek_identical_global}
Under the assumptions and notation of Theorem \ref{thm:Klvopotek_k_means_projected_cost},
if the data set $Q$ has the property of $\sqrt{s}$-Multiplicative Perturbation Robustness
with $s<1-\delta$, and if $\mathfrak{C_G}$ is the global optimum of $k$-means in $Q$, 
then it is also the global optimum in $Q'$ 
with probability at least $1-\epsilon$
\end{lemma}
\begin{proof}
Assume the contrary that is that in $Q'$ some other  clustering $\mathfrak{C'_G}$ is the global optimum. 
Let us define the distance 
$\sqrt{d_1(i,j)}=\|\mathbf{x_i}-\mathbf{x_j}\|$ and
$\sqrt{d_2(i,j)}=\frac{n}{n'}\|\mathbf{x_i'}-\mathbf{x_j'}\|$.
The distance $\sqrt{d_2}$ is a realistic distance in the coordinate system $\mathcal{C}$ as we assume $n>n'$.
As the $k$-means optimum does not change under rescaling, 
so $\mathfrak{C'_G}$ is also an optimal solution for clustering task under $d_2$.
But 
$$s d_1(i,j)<(1-\delta) d_1(i,j) \le  d_2(i,j) 
\le (1+\delta) d_1(i,j) <(1-\delta)^{-1} d_1(i,j) < s^{-1}  d_1(i,j)$$
hence the distance $\sqrt{d_2}$ is a perturbation of $\sqrt{d_1}$ 
and hence $\mathfrak{C_G}$ should be optimal under $\sqrt{d_2}$ also. 
We get a contradiction. So the claim of the lemma must be true.
\end{proof}

This implies that  
\begin{theorem}\label{thm:clusterability-multPertRobust}
Under the assumptions and notation of Theorem \ref{thm:Klvopotek_k_means_projected_cost},
if the data set $Q$ has the property of $\sqrt{s}$-Multiplicative Perturbation Robustness
with factor $s< s_p\nu  \frac{(1-\delta)^2}{1+\delta}$ ($0<s_p,\nu <1$)  in the original space, 
then 
with probability at least $1-2\epsilon$
it has the property of $\sqrt{s_p}$-Multiplicative Perturbation Robustness in the projected space.
\end{theorem}
\begin{proof}
The Lemma \ref{lem:Klvopotek_identical_global} implies that the global 
optima of the original and projected spaces are identical. 
So assume that in the original space 
for the distance 
$\sqrt{d_1(i,j)}=\|\mathbf{x_i}-\mathbf{x_j}\|$
$\mathfrak{C_G}$ is the optimal clustering.
Then under projection 
$\sqrt{d_1'(i,j)}=\|\mathbf{x_i'}-\mathbf{x_j'}\|$ 
we have the same  optimal clustering. 

For a perturbation with factor $s_p$ in the projected space define the distance 
$\sqrt{d_2'(i,j)}=\|\mathbf{y_i'}-\mathbf{y_j'}\|$ where for any $i$ let $\mathbf{y_i'}$ be a perturbation of $\mathbf{x_i'}$.
We will be done if we can demonstrate that 
$\sqrt{d_2'}$ yields the same optimum in the projected space as $\sqrt{d_1'}$ does. 
For any $i$ let   $\mathbf{y_i}$ be some point in the original space such that 
   $\mathbf{y_i'}$ is its projection to the projected space. 
We will treat $\mathbf{y_i}$ as an image of $\mathbf{x_i}$ and will subsequently show that the set of these points $\mathbf{y_i}$ can be treated as a perturbation of  $\mathbf{x_i}$ with the factor $\sqrt{s}$. 

For each counterpart $\sqrt{d_2(i,j)}=\|\mathbf{y_i}-\mathbf{y_j}\|$ of $\sqrt{d_2'}$  in original space 
$(1+\delta)^{-1}\frac{n}{n'}d_2'(i,j)\le d_2(i,j)\le (1-\delta)^{-1}\frac{n}{n'}d_2'(i,j)$ 
holds. 
As 
$s_p d_1'(i,j)\le d_2'(i,j)\le (s_p)^{-1}d_1'(i,j)$ 
and
$(1-\delta) d_1 (i,j)\le \frac{n}{n'}d'_1(i,j)\le (1+\delta) d_1(i,j)$ we obtain
$$s  d_1 (i,j) < (1+\delta)^{-1}s_p(1-\delta)  d_1 (i,j)\le
(1+\delta)^{-1}s_p \frac{n}{n'}d'_1(i,j)
\le(1+\delta)^{-1} \frac{n}{n'}d'_2(i,j)
$$ $$\le d_2(i,j)
\le(1-\delta)^{-1} \frac{n}{n'}d'_2(i,j)
$$ $$\le(1-\delta)^{-1}s_p^{-1} \frac{n}{n'}d'_1(i,j)
\le(1-\delta)^{-1}s_p^{-1}(1+\delta)  d_1 (i,j) < \frac 1s  d_1 (i,j)   
$$

So $\sqrt{d_2}$ is  a perturbation of $\sqrt{d_1}$ with the factor $\sqrt{s}$.
$\sqrt{d_1}$  is $\sqrt{s}$-multiplicative perturbation robust, therefore    both have the same optimal solution $\mathfrak{C_G}$.
Furthermore $\sqrt{d_2}$ 
has the property of $\sqrt{\nu (1-\delta)}$-Multiplicative Robustness (see Lemma \ref{lem:Klvopotek_double_perturbation}).
Therefore its counterpart $\sqrt{d_2'}$ has the same optimum clustering $\mathfrak{C_G}$ as $\sqrt{d_2}$
(see Lemma \ref{lem:Klvopotek_identical_global}), hence as $\sqrt{d_1}$, hence as $\sqrt{d_1'}$. 

Recall that $\sqrt{d_2'}$ was selected as any perturbation of $\sqrt{d_1'}$ with factor $\sqrt{s_p}$.
And it turned out that it yields the same optimal solution as $\sqrt{d_1'}$. 
So with high probability (factor 2 is taken as we deal with two data sets, comprising points $\mathbf{x}_i$ and $\mathbf{y}_i$)
$\sqrt{d_1'}$ possesses 
$\sqrt{s_p}$-Multiplicative Perturbation Robustness in the projected space.
\end{proof}

We claim 
\begin{theorem}\label{thm:clusterability-betaCentreStability}
Under the assumptions and notation of Theorem \ref{thm:Klvopotek_k_means_projected_cost},
if the data set $Q$ has both the property of 
$\beta$-Centre Stability
  and  $\sqrt{s}$-Multiplicative Perturbation Robustness
with $s<1-\delta$ in the original space, 
then 
with probability at least $1-\epsilon$
it has the property of 
$\beta\sqrt{ \frac{1-\delta}{1+\delta}}$-Centre Stability in the projected space. 
\end{theorem}
\begin{proof}
The   $\sqrt{s}$-Multiplicative Perturbation Robustness ensures that both the original and the projected space share same optimal clustering $\mathfrak{C}$. 
\\
Consider a data point $\mathbf{x_i}$ and a cluster $C\in \mathfrak{C}$ not containing $i$.
Then $\mathbf{x_i}$, $\boldsymbol\mu(C)$ and  $\boldsymbol\mu(C\cup\{i\})$
are colinear.
So are $\mathbf{x_i'}$, $\boldsymbol\mu'(C)$ and  $\boldsymbol\mu'(C\cup\{i\})$, that is the respective (linear) projections. 
Furthermore 
$\frac{
\|\mathbf{x_i} -\boldsymbol\mu(C\cup\{i\})\|
}{
\| \boldsymbol\mu(C)-\boldsymbol\mu(C\cup\{i\})\|
}=\frac{|C|}{1}$,
hence 
$\|\mathbf{x_i}- \boldsymbol\mu(C) \|=
\frac{|C|+1}{|C|}
\|\mathbf{x_i} -\boldsymbol\mu(C\cup\{i\})\|
$.
Likewise 
$\frac{
\|\mathbf{x_i'} -\boldsymbol\mu'(C\cup\{i\})\|
}{
\| \boldsymbol\mu'(C)-\boldsymbol\mu'(C\cup\{i\})\|
}=\frac{|C|}{1}$.
Upon projection the distance  to own cluster centre can increase relatively by $\sqrt{1+\delta}$
 and to the $C\cup\{i\}$ centre 
can decrease by  $\sqrt{1-\delta}$, see Lemma \ref{lem:Klvopotek_Lemma_point_cluster_centre}.
That means 
$ \|\mathbf{x_i'}-\boldsymbol\mu'(\mathfrak{C}(i))\|^2 \le (1+\delta)  \frac{n'}{n} 
 \|\mathbf{x_i}-\boldsymbol\mu(\mathfrak{C}(i))\|^2
$ and 
$ (1-\delta)^{-1} \|\mathbf{x_i'}-\boldsymbol\mu'(C\cup\{i\})\|^2 \ge  \frac{n'}{n} 
 \|\mathbf{x_i}-\boldsymbol\mu(C\cup\{i\})\|^2
$.
Due to the aforementioned relations 
$ \|\mathbf{x_i'}-\boldsymbol\mu'(C )\|^2 \ge (1-\delta) \frac{n'}{n} 
 \|\mathbf{x_i}-\boldsymbol\mu(C )\|^2
$.
Due to $\beta$-Centre-Stability in the original space we had:
$ \beta^2  \|\mathbf{x_i}-\boldsymbol\mu(\mathfrak{C}(i))\|^2
<  \|\mathbf{x_i}-\boldsymbol\mu(C )\|^2
$.
Due to the aforementioned relations we have 
$$ \|\mathbf{x_i'}-\boldsymbol\mu'(C )\|^2 \ge (1-\delta) \frac{n'}{n} 
 \|\mathbf{x_i}-\boldsymbol\mu(C )\|^2
$$ $$
> \beta^2 (1-\delta) \frac{n'}{n}  \|\mathbf{x_i}-\boldsymbol\mu(\mathfrak{C}(i))\|^2
\ge \beta^2 \frac{1-\delta}{1+\delta} \|\mathbf{x_i'}-\boldsymbol\mu'(\mathfrak{C}(i))\|^2
$$
That is $\|\mathbf{x_i'}-\boldsymbol\mu'(C )\|
  > 
 \beta \sqrt{ \frac{1-\delta}{1+\delta}} \|\mathbf{x_i'}-\boldsymbol\mu'(\mathfrak{C}(i))\|$ 
Hence the data centre stability can drop to 
 $\beta\sqrt{ \frac{1-\delta}{1+\delta}} $.
\end{proof}

We claim 
\begin{theorem}\label{thm:clusterability-weakDeletionStability}
Under the assumptions and notation of Theorem \ref{thm:Klvopotek_k_means_projected_cost},
if the data set $Q$ has both the property of 
$(1+\beta)$ Weak Deletion Stability
  and  $\sqrt{s}$-Multiplicative Perturbation Robustness
with $s<1-\delta$ in the original space, 
then 
with probability at least $1-\epsilon$
it has the property of 
$(1+\beta)  \frac{ 1-\delta}{1+\delta}$ Weak Deletion Stability in the projected space. 
\end{theorem}
\begin{proof}
The   $\sqrt{s}$-Multiplicative Perturbation Robustness ensures that both original and the projected space share same optimal clustering. 
Let this optimal clustering be called $\mathfrak{C_o}$. 
By $\mathfrak{C}$ denote any clustering obtained from 
$\mathfrak{C_o}$ by deletion of one cluster centre and assigning cluster elements to one of the remaining clusters. 
By the assumption of (1+$\beta$)-Weak Deletion stability
$(1+\beta) \mathfrak{J}(Q,\mathfrak{C_o}) \le \mathfrak{J}(Q,\mathfrak{C}) $.
\\
Theorem  \ref{thm:Klvopotek_k_means_projected_cost} 
implies that 
$(1-\delta) \frac{n'}{n} \mathfrak{J}(Q,\mathfrak{C}) \le 
\mathfrak{J}(Q',\mathfrak{C})$
and $ (1+\delta)^{-1}\mathfrak{J}(Q',\mathfrak{C_o})
  \le  \frac{n'}{n}  \mathfrak{J}(Q,\mathfrak{C_o})$. 
\\
Therefore 
$\mathfrak{J}(Q',\mathfrak{C})
\ge (1-\delta) \frac{n'}{n}\mathfrak{J}(Q,\mathfrak{C})
\ge (1+\beta)  (1-\delta) \frac{n'}{n} \mathfrak{J}(Q,\mathfrak{C_o}) 
\ge (1+\beta)  (1-\delta)   (1+\delta)^{-1} \mathfrak{J}(Q',\mathfrak{C_o}) $
which implies the claim. 
\end{proof}


\begin{table} 
\caption{Dependence of reduced dimensionality $n'$ on sample size $m$. Other parameters fixed at  $\epsilon$=0.01 $\delta$=0.05 $n$=5e+05.}
\label{tab:samplesize}
\begin{center}  
\begin{tabular}{|r|r|r|r|} 
\hline  
$m$ & $n'$ explicit &$n'$ implicit & explicit/implicit \\ 
\hline  
  10 &  15226 &  14209 &  1.07 \\
   20 &  17518 &  16389 &  1.07 \\
   50 &  20547 &  19191 &  1.07 \\
   100 &  22839 &  21269 &  1.07 \\
   200 &  25131 &  23323 &  1.08 \\
   500 &  28160 &  26016 &  1.08 \\
   1000 &  30452 &  28030 &  1.09 \\
   2000 &  32744 &  30027 &  1.09 \\
   5000 &  35773 &  32648 &  1.1 \\
   10000 &  38065 &  34609 &  1.1 \\
   20000 &  40357 &  36554 &  1.1 \\
   50000 &  43386 &  39097 &  1.11 \\
   1e+05 &  45678 &  41017 &  1.11 \\
   2e+05 &  47970 &  42910 &  1.12 \\
   5e+05 &  50999 &  45392 &  1.12 \\
   1e+06 &  53291 &  47250 &  1.13 \\
   2e+06 &  55582 &  49099 &  1.13 \\
   5e+06 &  58612 &  51515 &  1.14 \\
   1e+08 &  68516 &  59243 &  1.16 \\
   2e+07 &  63195 &  55127 &  1.15 \\
   5e+07 &  66225 &  57480 &  1.15 \\
   1e+08 &  68516 &  59243 &  1.16 \\
\hline  
\end{tabular}  
\end{center}  
\end{table}

\begin{figure} 
\centering
\includegraphics[width=0.8\textwidth]{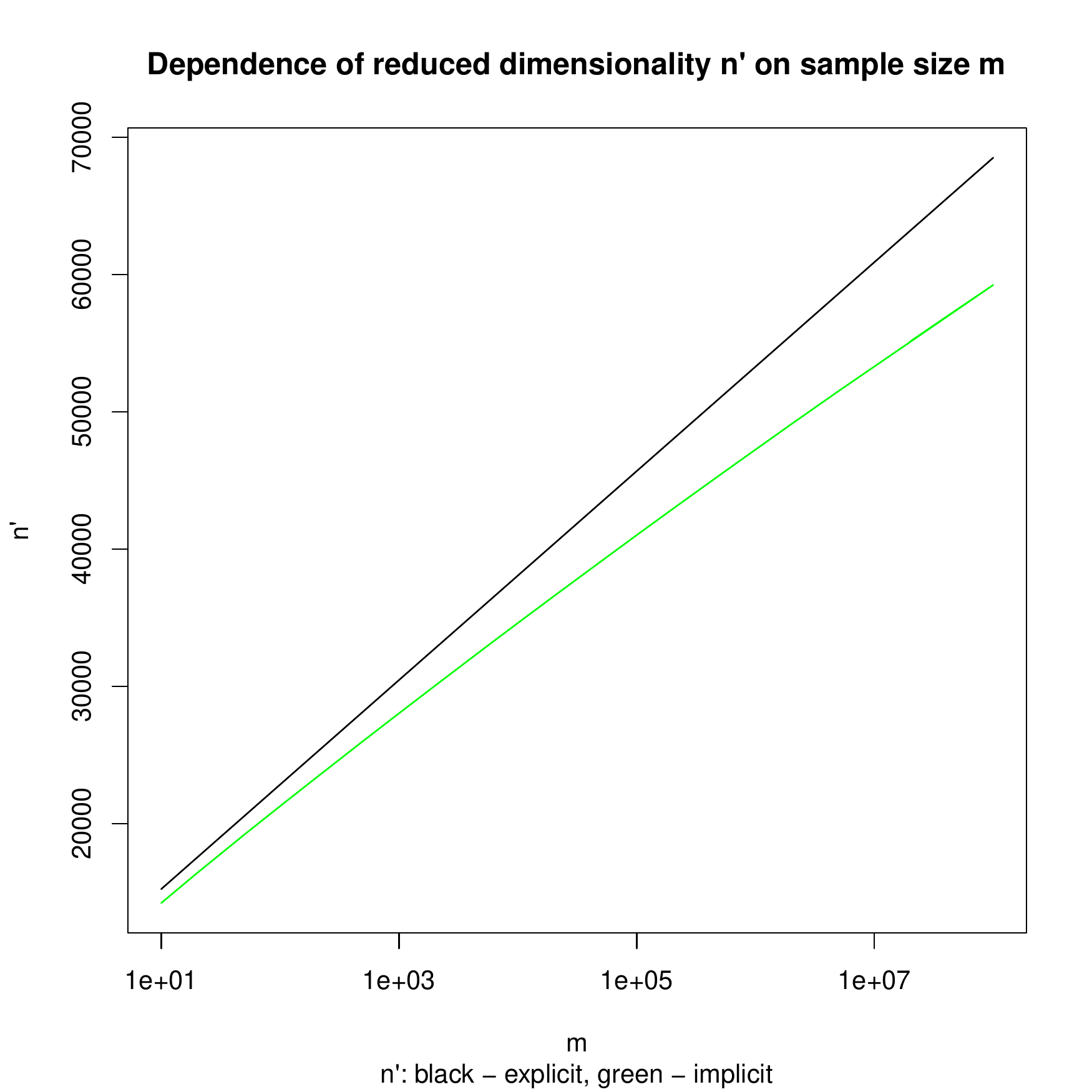} %
\caption{Dependence of reduced dimensionality $n'$ on sample size $m$. Other parameters fixed at  $\epsilon$=0.01 $\delta$=0.05 $n$=5e+05}
\label{fig:samplesize}
\end{figure}

\begin{table} 
\caption{Dependence of reduced dimensionality $n'$ on failure prob. $\epsilon$. Other parameters fixed at  $m$=2e+06 $\delta$=0.05 $n$=5e+05.}
\label{tab:epsilon}
\begin{center}  
\begin{tabular}{|r|r|r|r|} 
\hline  
$\epsilon$ & $n'$ explicit &$n'$ implicit & explicit/implicit \\ 
\hline  
  0.1 &  51776 &  46020 &  1.13 \\
   0.05 &  52922 &  46955 &  1.13 \\
   0.02 &  54437 &  48180 &  1.13 \\
   0.01 &  55582 &  49099 &  1.13 \\
   0.005 &  56728 &  50014 &  1.13 \\
   0.002 &  58243 &  51221 &  1.14 \\
   0.001 &  59389 &  52134 &  1.14 \\
\hline  
\end{tabular}  
\end{center}  
\end{table}

\begin{figure} 
\centering
\includegraphics[width=0.8\textwidth]{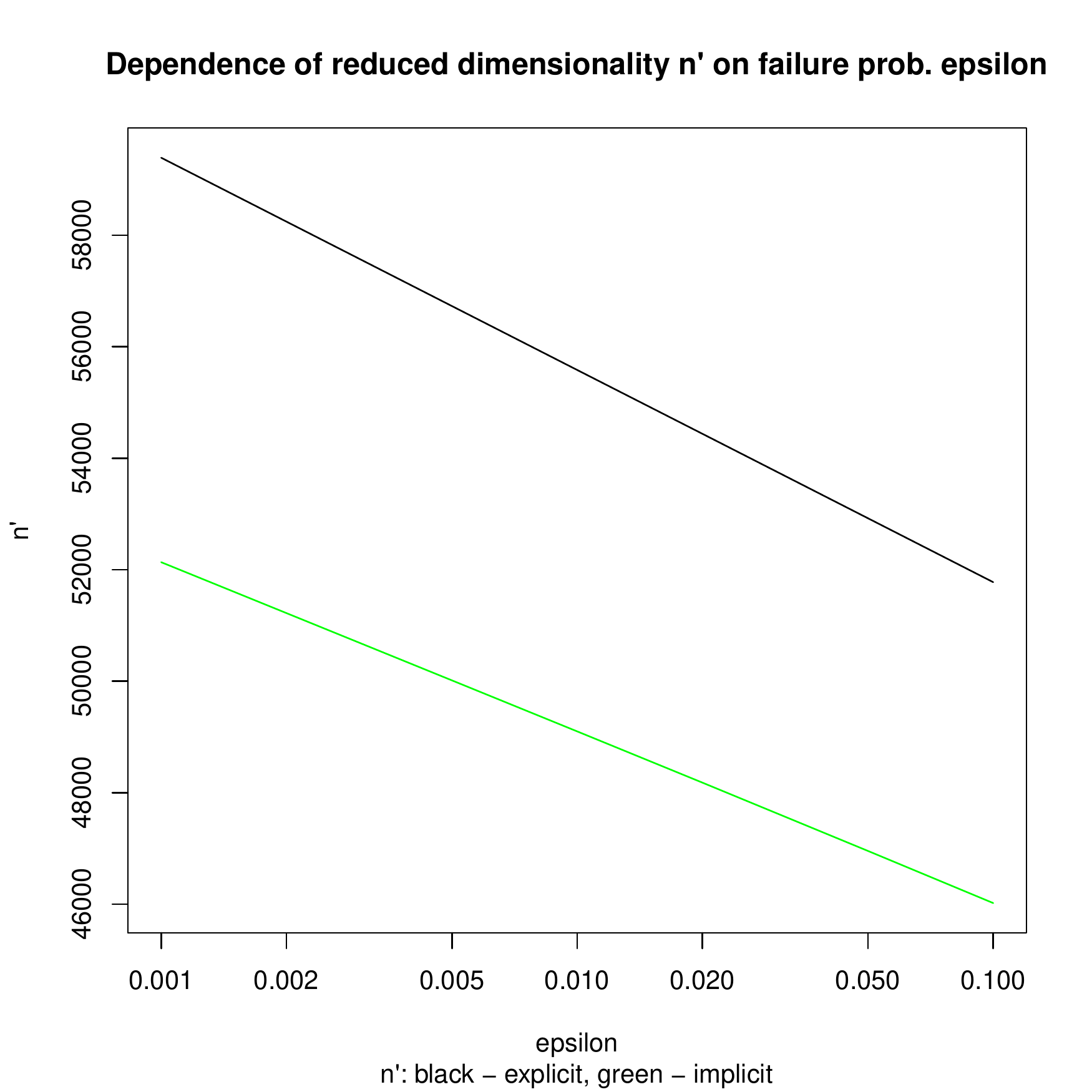} %
\caption{Dependence of reduced dimensionality $n'$ on failure prob. $\epsilon$. Other parameters fixed at  $m$=2e+06 $\delta$=0.05 $n$=5e+05}
\label{fig:epsilon}
\end{figure}

\begin{table} 
\caption{Dependence of reduced dimensionality $n'$ on error range $\delta$. Other parameters fixed at  $m$=2e+06 $\epsilon$=0.01 $n$=5e+05.}
\label{tab:delta}
\begin{center}  
\begin{tabular}{|r|r|r|r|} 
\hline  
$\delta$ & $n'$ explicit &$n'$ implicit & explicit/implicit \\ 
\hline  
  0.5 &  712 &  697 &  1.02 \\
   0.4 &  1059 &  1032 &  1.03 \\
   0.3 &  1787 &  1745 &  1.02 \\
   0.2 &  3804 &  3692 &  1.03 \\
   0.1 &  14339 &  13640 &  1.05 \\
   0.09 &  17593 &  16631 &  1.06 \\
   0.08 &  22128 &  20742 &  1.07 \\
   0.07 &  28721 &  26604 &  1.08 \\
   0.06 &  38846 &  35329 &  1.1 \\
   0.05 &  55582 &  49099 &  1.13 \\
   0.04 &  86291 &  72387 &  1.19 \\
   0.03 &  152415 &  115298 &  1.32 \\
   0.02 &  340701 &  201059 &  1.69 \\
   0.01 &  1353858 &  1353859 &  1 \\
\hline  
\end{tabular}  
\end{center}  
\end{table}

\begin{figure} 
\centering
\includegraphics[width=0.8\textwidth]{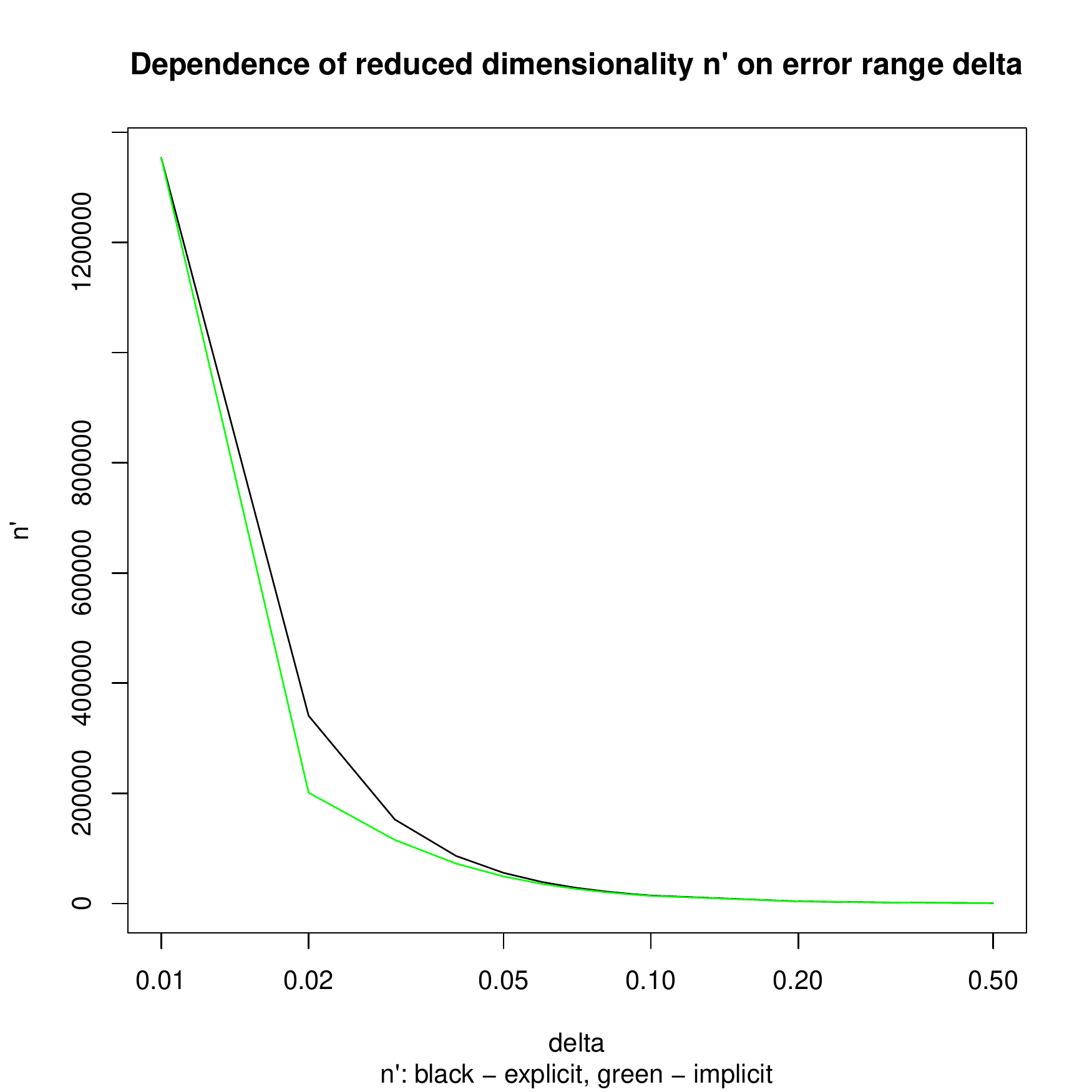} %
\caption{Dependence of reduced dimensionality $n'$ on error range $\delta$. Other parameters fixed at  $m$=2e+06 $\epsilon$=0.01 $n$=5e+05}
\label{fig:delta}
\end{figure}

\begin{table} 
\caption{Dependence of reduced dimensionality $n'$ on original dimensionality $n$. Other parameters fixed at  $m$=2e+06 $\epsilon$=0.01 $\delta$=0.05.}
\label{tab:orign}
\begin{center}  
\begin{tabular}{|r|r|r|r|} 
\hline  
$n$ & $n'$ explicit &$n'$ implicit & explicit/implicit \\ 
\hline  
  4e+05 &  55582 &  47891 &  1.16 \\
   5e+05 &  55582 &  49099 &  1.13 \\
   6e+05 &  55582 &  49933 &  1.11 \\
   7e+05 &  55582 &  50551 &  1.1 \\
   8e+05 &  55582 &  51025 &  1.09 \\
   9e+05 &  55582 &  51399 &  1.08 \\
   1e+06 &  55582 &  51703 &  1.08 \\
\hline  
\end{tabular}  
\end{center}  
\end{table}

\begin{figure} 
\centering
\includegraphics[width=0.8\textwidth]{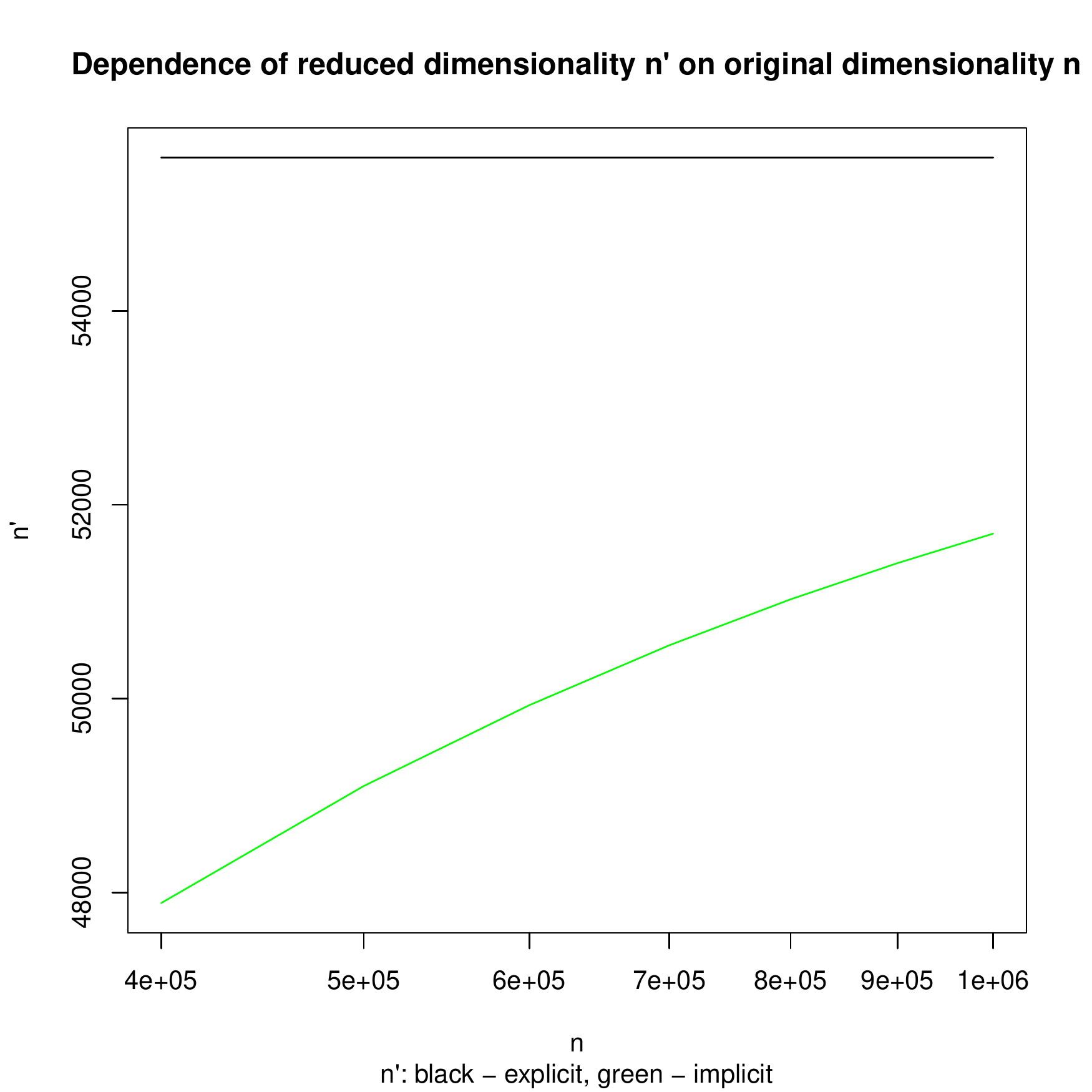} %
\caption{Dependence of reduced dimensionality $n'$ on original dimensionality $n$. Other parameters fixed at  $m$=2e+06 $\epsilon$=0.01 $\delta$=0.05}
\label{fig:orign}
\end{figure} 

\begin{figure} 
\centering
\includegraphics[width=0.8\textwidth]{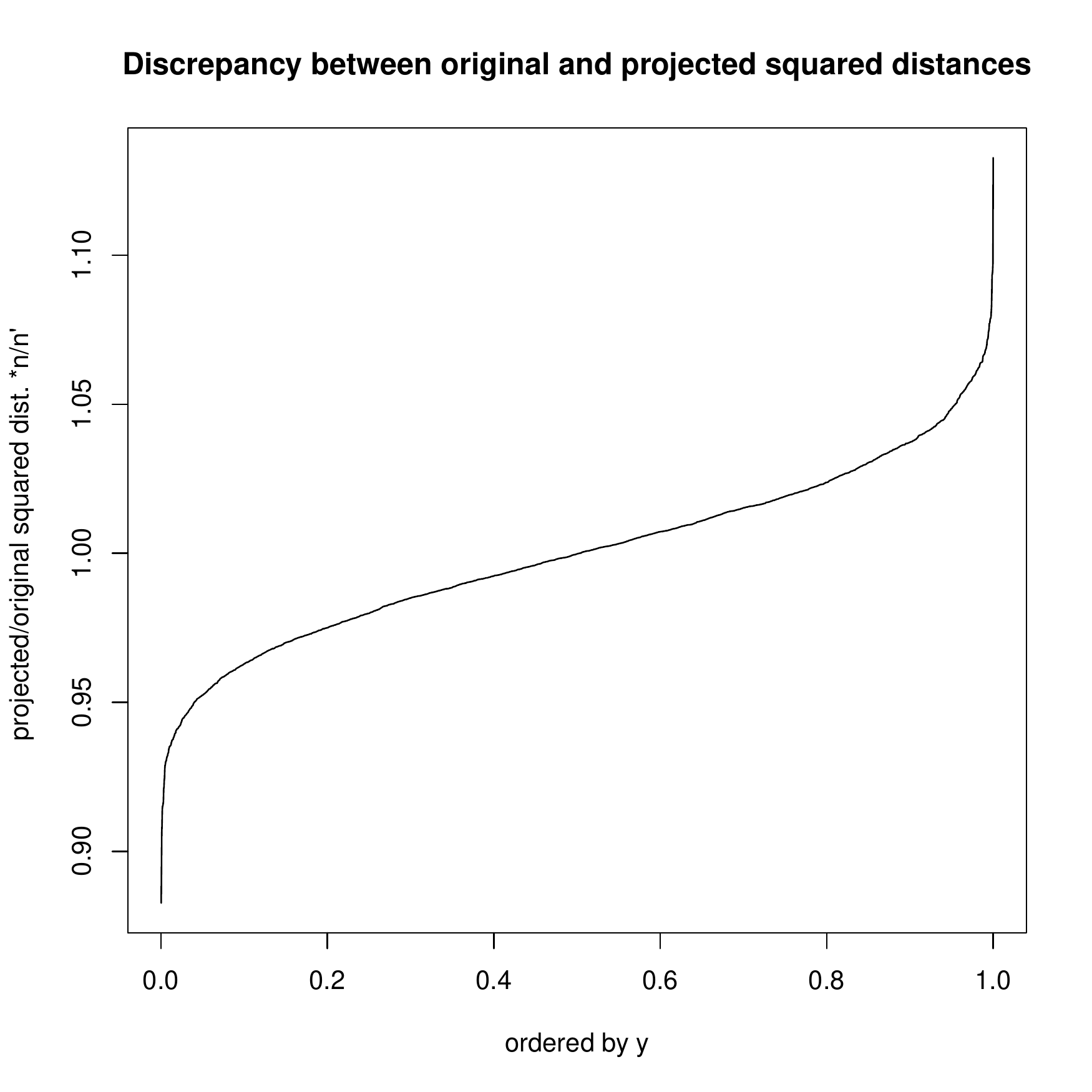} %
\caption{Discrepancy between   projected and original squared distances between points in the sample expressed as their quotient adjusted by $n/n'$. Parameters fixed at  m= 5000  $\epsilon$= 0.1  $\delta$= 0.2  $n$= 5000  $n'$= 2188}
\label{fig:distortion}
\end{figure} 


\section{Numerical Experiments on Some Aspects of Our Approach}\label{sec:examples}

Note that we have two formulas for computing the reduced space dimensionality  $n'$,
the formula (\ref{eq:nprime_implicit}) and (\ref{eq:nprime_computation}).  
The latter does not engage the original dimensionality $n$,
while it is explicit in $n'$. The value of $n'$ in the former depends on $n$, however $n'$ can be only computed iteratively.

Let us investigate the differences between $n'$ computation in both cases. 
Let us check the impact of the following parameters:  
$n$ - the original dimensionality (see table \ref{tab:orign} and figure \ref{fig:orign}),
$\delta$ - the limitation of deviation of the distances between data points in the original and the reduced space  (see table \ref{tab:delta} and figure \ref{fig:delta}),
$m$ - the sample size  (see table \ref{tab:samplesize} and figure \ref{fig:samplesize}), as well as 
$\epsilon$ - the maximum failure probability of the "JL" transformation   (see table \ref{tab:epsilon} and figure \ref{fig:epsilon}). 
Note that in all figures the X-axis is on log scale. 

As visible in figure \ref{fig:orign} the value of $n'$ from the explicit formula does not depend on the original dimensionality $n$.
The value computed from the implicit formula approaches the explicit value quite quickly with the growing dimensionality $n$.

On the other hand, the implicit $n'$ departs from the explicit one with growing sample size $m$, as visible in fig. \ref{fig:samplesize}. 
Both grow with increasing $m$. 

In fig. \ref{fig:epsilon} we see that when we increase the acceptable failure rate $\epsilon$, the requested dimensionality $n'$ drops, whereby the implicit one approaches the explicit one. 

Fig. \ref{fig:delta} shows that the requested dimensionality drops quite quickly with increased relative error range $\delta$ till a kind of saturation is achieved. 
At extreme ends of $\delta$ implicit and explicit $n'$ formulas converge to one another. 

The behaviour of explicit $n'$ is not surprising, as it is visible  directly from the formula (\ref{eq:nprime_computation}).
The important insight here is however the required dimensionality of the projected data, of hundreds of thousands for realistic $\epsilon, \delta$.  
So the random projection via the Johnson-Lindenstrauss Lemma is not yet another dimensionality reduction technique. 
It is suitable for cases where techniques like PCA are not feasible computationally. 

The behaviour of implicit $n'$ for the case of increasing original dimensionality $n$ is as expected - the explicit $n'$ reflects the "in the limit" behaviour of the implicit formulation.
The convergence for extreme values of $\delta$ is intriguing.
The discrepancy for $\epsilon$ and the divergence for growing $m$ indicate 
that there is still space for better explicit formulas on $n'$. 
Especially it is worth investigating for increasing $m$ as the processing becomes more expensive in the original space when $m$ is increasing. 

In order to give an impression how effective the random projection is,
see fig. \ref{fig:distortion}.
It illustrates the distribution of discrepancies between squared distances 
in the projected and in the original spaces. 
The discrepancies are expressed as 
  $$\frac{ \|f(\mathbf{u})-f(\mathbf{v})\|^2}{ \| \mathbf{u}-\mathbf{v}\|^2 }$$
One can see that they correspond quite well to the imposed constraints. 


\begin{figure} 
\centering
\includegraphics[width=0.8\textwidth]{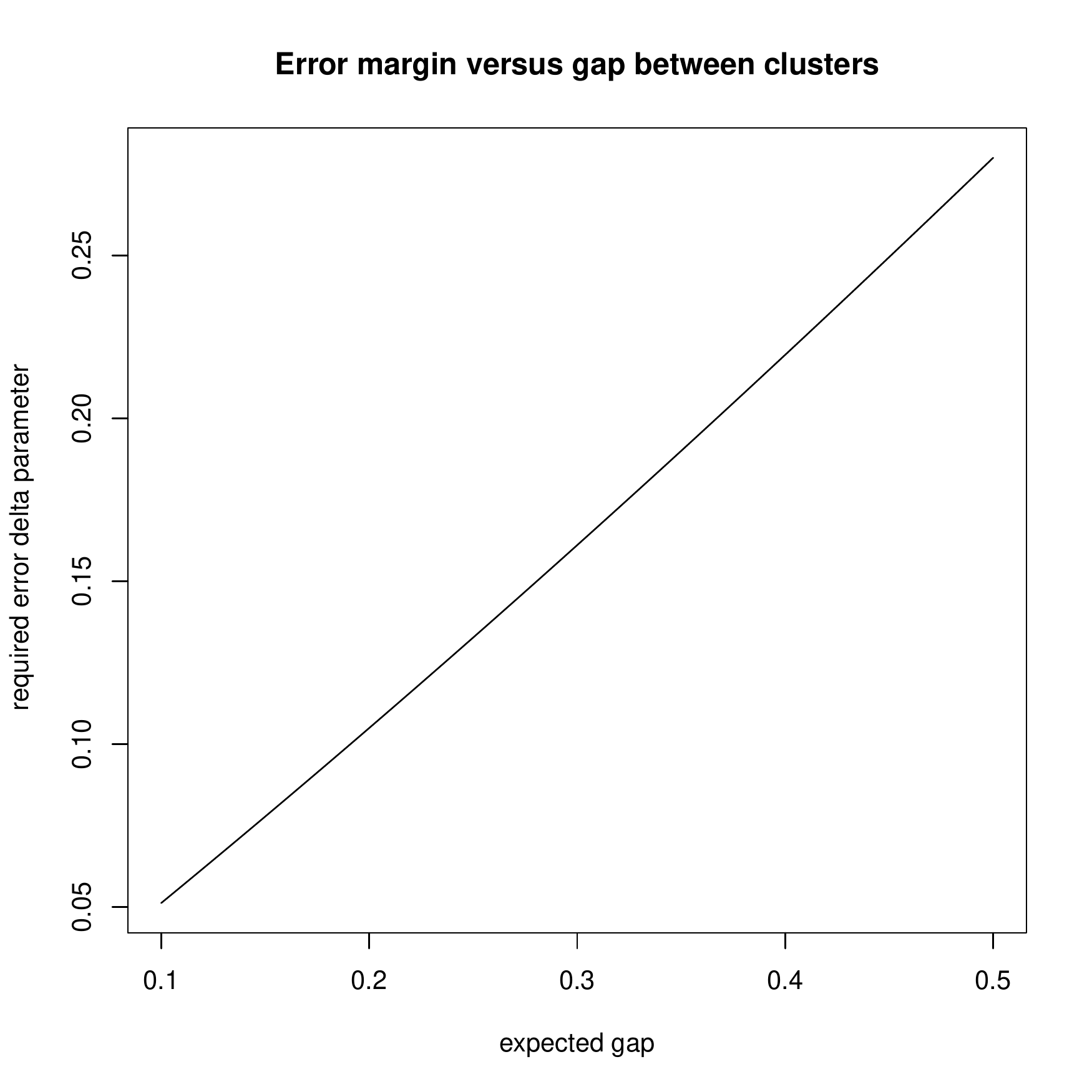} %
\caption{Permissible error range $\delta$ under various assumed gaps between the clusters.}
\label{fig:gap_delta}
\end{figure} 

As the application for $k$-means clustering, we see in fig. 
 \ref{fig:gap_delta} that the bigger the relative gap between clusters, the larger the error value $\delta$ is permitted, if class membership shall not be distorted by the projection.


\begin{table} 
\caption{Comparison of effort needed for $k$-means under our 
dimensionality reduction approach and that of 
Dasgupta and Gupta \cite{Dasgupta:2003}, depending on   sample size $m$ . Other parameters fixed at  $\epsilon$= 0.01  $\delta$= 0.05  $n$= 5e+05}
\label{tab:DGcmpsamplesize}
\begin{center}  
\begin{tabular}{|r|r|r|r|r|r|} 
\hline  
$m$ & $n'$ explicit &$n'$ implicit & Gupta $n'$ & Their Repetitions &Our $n'$ to their $n'$\\ 
\hline  
  10 &  15226 &  14209 &  3879 &  44 &  3.7 \\
   20 &  17518 &  16389 &  5046 &  90 &  3.3 \\
   50 &  20547 &  19191 &  6589 &  228 &  3 \\
   100 &  22839 &  21269 &  7757 &  459 &  2.8 \\
   200 &  25131 &  23323 &  8924 &  919 &  2.7 \\
   500 &  28160 &  26016 &  10467 &  2301 &  2.5 \\
   1000 &  30452 &  28030 &  11635 &  4603 &  2.5 \\
   2000 &  32744 &  30027 &  12802 &  9209 &  2.4 \\
   5000 &  35773 &  32648 &  14345 &  23024 &  2.3 \\
   10000 &  38065 &  34609 &  15513 &  46050 &  2.3 \\
   20000 &  40357 &  36554 &  16680 &  92102 &  2.2 \\
   50000 &  43386 &  39097 &  18223 &  230257 &  2.2 \\
   1e+05 &  45678 &  41017 &  19391 &  460515 &  2.2 \\
   2e+05 &  47970 &  42910 &  20558 &  921032 &  2.1 \\
   5e+05 &  50999 &  45392 &  22101 &  2302583 &  2.1 \\
   1e+06 &  53291 &  47250 &  23269 &  4605168 &  2.1 \\
   2e+06 &  55582 &  49099 &  24436 &  9210339 &  2.1 \\
   5e+06 &  58612 &  51515 &  25979 &  23025849 &  2 \\
   1e+08 &  68516 &  59243 &  31025 &  460517014 &  2 \\
   2e+07 &  63195 &  55127 &  28314 &  92103402 &  2 \\
   5e+07 &  66225 &  57480 &  29857 &  230258508 &  2 \\
   1e+08 &  68516 &  59243 &  31025 &  460517014 &  2 \\
\hline  
\end{tabular}  
\end{center}  
\end{table}

\begin{table} 
\caption{Comparison of effort needed for $k$-means under our 
dimensionality reduction approach and that of 
Dasgupta and Gupta \cite{Dasgupta:2003}, depending on   failure prob. $\epsilon$ . Other parameters fixed at  m= 2e+06  $\delta$= 0.05  $n$= 5e+05}
\label{tab:DGcmpepsilon}
\begin{center}  
\begin{tabular}{|r|r|r|r|r|r|} 
\hline  
$\epsilon$ & $n'$ explicit &$n'$ implicit & Gupta $n'$ & Their Repetitions &Our $n'$ to their $n'$\\ 
\hline  
  0.1 &  51776 &  46020 &  24436 &  4605170 &  1.9 \\
   0.05 &  52922 &  46955 &  24436 &  5991464 &  2 \\
   0.02 &  54437 &  48180 &  24436 &  7824045 &  2 \\
   0.01 &  55582 &  49099 &  24436 &  9210339 &  2.1 \\
   0.005 &  56728 &  50014 &  24436 &  10596633 &  2.1 \\
   0.002 &  58243 &  51221 &  24436 &  12429214 &  2.1 \\
   0.001 &  59389 &  52134 &  24436 &  13815508 &  2.2 \\
\hline  
\end{tabular}  
\end{center}  
\end{table}

\begin{table} 
\caption{Comparison of effort needed for $k$-means under our 
dimensionality reduction approach and that of 
Dasgupta and Gupta \cite{Dasgupta:2003}, depending on   error range $\delta$ . Other parameters fixed at  m= 2e+06  $\epsilon$= 0.01  $n$= 5e+05}
\label{tab:DGcmpdelta}
\begin{center}  
\begin{tabular}{|r|r|r|r|r|r|} 
\hline  
$\delta$ & $n'$ explicit &$n'$ implicit & Gupta $n'$ & Their Repetitions &Our $n'$ to their $n'$\\ 
\hline  
  0.5 &  712 &  697 &  465 &  9210339 &  1.5 \\
   0.4 &  1059 &  1032 &  605 &  9210339 &  1.8 \\
   0.3 &  1787 &  1745 &  922 &  9210339 &  1.9 \\
   0.2 &  3804 &  3692 &  1814 &  9210339 &  2.1 \\
   0.1 &  14339 &  13640 &  6449 &  9210339 &  2.2 \\
   0.09 &  17593 &  16631 &  7874 &  9210339 &  2.2 \\
   0.08 &  22128 &  20742 &  9857 &  9210339 &  2.2 \\
   0.07 &  28721 &  26604 &  12736 &  9210339 &  2.1 \\
   0.06 &  38846 &  35329 &  17150 &  9210339 &  2.1 \\
   0.05 &  55582 &  49099 &  24436 &  9210339 &  2.1 \\
   0.04 &  86291 &  72387 &  37783 &  9210339 &  2 \\
   0.03 &  152415 &  115298 &  66478 &  9210339 &  1.8 \\
   0.02 &  340701 &  201059 &  148048 &  9210339 &  1.4 \\
   0.01 &  1353858 &  1353859 &  586209 &  9210339 &  2.4 \\
\hline  
\end{tabular}  
\end{center}  
\end{table}

\begin{table} 
\caption{Comparison of effort needed for $k$-means under our 
dimensionality reduction approach and that of 
Dasgupta and Gupta \cite{Dasgupta:2003}, depending on   original dimensionality $n$ . Other parameters fixed at  m= 2e+06  $\epsilon$= 0.01  $\delta$= 0.05}
\label{tab:DGcmporign}
\begin{center}  
\begin{tabular}{|r|r|r|r|r|r|} 
\hline  
$n$ & $n'$ explicit &$n'$ implicit & Gupta $n'$ & Their Repetitions &Our $n'$ to their $n'$\\ 
\hline  
  4e+05 &  55582 &  47891 &  24436 &  9210339 &  2 \\
   5e+05 &  55582 &  49099 &  24436 &  9210339 &  2.1 \\
   6e+05 &  55582 &  49933 &  24436 &  9210339 &  2.1 \\
   7e+05 &  55582 &  50551 &  24436 &  9210339 &  2.1 \\
   8e+05 &  55582 &  51025 &  24436 &  9210339 &  2.1 \\
   9e+05 &  55582 &  51399 &  24436 &  9210339 &  2.2 \\
   1e+06 &  55582 &  51703 &  24436 &  9210339 &  2.2 \\
\hline  
\end{tabular}  
\end{center}  
\end{table} 


\section{Previous work}\label{sec:previousWork}

Note that if we would set $\epsilon$ (close) to 1, 
and expand by Taylor method the $\ln$ function in denominator
of the inequality (\ref{eq:nprime_computation}) 
 to up to three terms
then we get the value of $n'$ from equation (2.1) from the paper \cite{Dasgupta:2003}: 
$$n'\ge 4\frac{\ln m}{\delta^2-\delta^3}$$

 Note, however, that setting $\epsilon$ to a value close to 1 does not make sense as we want to keep rare the event that the data does not fit the interval we are imposing.

Though one may be tempted to view our results as formally similar to those of Dasgupta and Gupta, there is one major difference. 
Let us first recall that the original proof of Johnson and Lindenstrauss~\cite{Johnson:1982} is probabilistic, showing that projecting the
$m$ -point subset onto a random subspace of
$O
(\ln
m
/\epsilon^2
)$ dimensions only changes
the  (squared) distances between points by at most $1-\delta$   with positive probability.
Dasgupta and Gupta showed that this probability is at least $1/m$, which is not much indeed. 
In order to get failure probability $\epsilon$ below say 0.05\%,
one needs to repeat the random projection and checking of distances 
$r$ times, with such $r$ that 
$\epsilon > (1-\frac 1m)^r$. 
In case of $m=1,000$ this means 
over $r=2,995$  repetitions, 
and with $m=1,000,000$ - over $r=2,995,000$  repetitions,

In this paper 
we have shown that this success probability can be raised to $1-\epsilon$ for an $\epsilon$ given in advance.
Hereby the increase of target dimensionality is small enough compared to Dasgupta and Gupta formula, 
that our random projection method is orders of magnitude more efficient. 
A detailed comparison is contained in the tables
\ref{tab:DGcmpsamplesize},
\ref{tab:DGcmpepsilon},
\ref{tab:DGcmpdelta},
\ref{tab:DGcmporign}.
We present in these tables $n'$ computed using our formulas with those proposed by Dasgupta and Gupta
as well as we present the required number of repetition of projection onto sampled subspaces in order to obtain a faithful distance discrepancies with reasonable probability. 
Dasgupta and Gupta generally  obtain several times lower number of dimensions.
However, as stated in the introduction, the number of repeated samplings annihilates this advantage and in fact a much higher burden when clustering is to be expected.

Note that the choice of $n'$ has been estimated by \cite{Achlioptas:2003} 
$$n'\ge (4+2\gamma)\frac{\ln m}{\delta^2-\delta^3}$$
 where $\gamma$ is some positive number. They propose a projection based on two or three discrete values randomly assigned  instead of ones from normal distribution. With the quantity $\gamma$ they control the probability that a single element of the set $Q$ leaves the predefined interval $\pm \delta$. 
They do not bother about controlling the probability that none of the elements leaves the interval of interest. Rather, they derive expected values of various moments. 

 Larsen and Nelson \cite{Larsen:2016} concentrate on finding the highest value of $n'$ for which Johnson-Lindenstrauss Lemma does not hold demonstrating that the value they found is the tightest even for non-linear mappings $f$. 
Though not directly related to our research, they discuss the other side of the coin, that is the dimensionality below which at least one point of the data set has to violate the constraints.

\section{Conclusions}\label{sec:conclusions}

In this paper we investigated a novel aspect of the well known and widely explored and exploited Johnson-Lindenstrauss lemma on the possibility of dimensionality reduction by projection onto a random subspace. 

The original formulation means in practice that we have to check whether or not we have found a proper transformation $f$ leading to error bounds within required range for all pairs of points, and if necessary (and it is theoretically necessary very frequently), to repeat the random projection process over and over again. 

We have shown here that it is possible to determine in advance the choice of dimensionality in the random projection process as to assure with desired certainty that none of the points of the data set violates restrictions on error bounds. 
This new formulation can be of importance for many data mining applications, like clustering, where the distortion of distances influences the results in a subtle way (e.g. $k$-means clustering).

Via some numerical examples we have pointed at the real application areas  of this kind of projections,   that is 
problems with high number of dimensions, starting with dozens of thousands  and hundreds of thousands of dimensions. 

Additionally, our reformulation of the JL Lemma permits to preserve some well-known clusterability properties at the projection.

\bibliographystyle{plain}
\bibliography{ProbabilisticJohnsonLindenstrauss_bib}
\end{document}